\documentclass[a4paper,USenglish,cleveref, autoref, thm-restate]{lipics-v2021}

\title{Fast and Practical Single-Exponential Algorithms for Branchwidth} 

\author{Taiki Kaneda}
{Hokkaido University, Sapporo, Japan}{kaneda.taiki.c9@elms.hokudai.ac.jp}{}{}
\author{Yasuaki Kobayashi}
{Hokkaido University, Sapporo, Japan}{koba@ist.hokudai.ac.jp}{https://orcid.org/0000-0003-3244-6915}{JSPS KAKENHI Grant numbers JP23K28034, JP24H00686, JP24H00697.}

\author{Hisao Tamaki}
{Meiji University, Kawasaki, Japan}{hisao.tamaki@gmail.com}{https://orcid.org/0000-0001-7566-8505}{JSPS KAKENHI Grant number JP24H00697.}

\authorrunning{T. Kaneda et al.} 

\Copyright{Taiki Kaneda, Yasuaki Kobayashi, and Hisao Tamaki} 



\category{} 

\relatedversion{} 

\nolinenumbers 

\EventEditors{John Q. Open and Joan R. Access}
\EventNoEds{2}
\EventLongTitle{42nd Conference on Very Important Topics (CVIT 2016)}
\EventShortTitle{CVIT 2016}
\EventAcronym{CVIT}
\EventYear{2016}
\EventDate{December 24--27, 2016}
\EventLocation{Little Whinging, United Kingdom}
\EventLogo{}
\SeriesVolume{42}
\ArticleNo{23}

\usepackage{mathtools}
\usepackage[ruled,linesnumbered,vlined]{algorithm2e}
\usepackage{todonotes}
\usepackage{booktabs}
\usepackage{float}

\newcommand{\bw}{\mathrm{bw}}
\newcommand{\bound}{\partial}
\newcommand{\midset}{\mathrm{mid}}
\newcommand{\lowset}{\mathcal{L}}
\newcommand{\upset}{\mathcal{U}}
\newcommand{\core}{{\kappa}}
\newcommand{\corehat}{{\hat\kappa}}

\newcommand{\closure}{\mathtt{clos}}
\newcommand{\bigOh}{\mathcal{O}}

\newcommand{\calB}{\mathcal{B}}
\newcommand{\calC}{\mathcal{C}}
\newcommand{\calF}{\mathcal{F}}
\newcommand{\calT}{\mathcal{T}}
\newcommand{\calD}{\mathcal{D}}

\newcommand{\calS}{\mathcal{S}}

\begin{document}

\maketitle

\begin{abstract}
In this paper, we present exact exponential algorithms for computing branchwidth that are fast both in theory and in practice. 
The running times of these algorithms are single-exponential in the number of vertices.
Our basic algorithm is based on a conceptually simple recurrence on vertex sets and computes the branchwidth of an $n$-vertex hypergraph in time $\bigOh^*(4^n)$. This is the first single-exponential time algorithm for hypergraphs.

We have two algorithms tailored specifically for graphs.
The first algorithm runs in time $\bigOh(3.293^n)$, improving upon the previously best-known running time of $\bigOh(3.4652^n)$ [Fomin-Mazoit-Todinca, DAM 2009]. 
Moreover, our computational experiment shows that it overwhelmingly outperforms state-of-the-art practical algorithms for computing branchwidth.
The second algorithm is a candidate for a theoretical improvement: we conjecture that it runs in time $\bigOh(c^n)$ for some constant $c$ that is smaller than 3.293.
In practice, it performs significantly better 
on some instances that are hard for the first algorithm.
\end{abstract}

\section{Introduction}
The \emph{branchwidth} of a graph is a well-studied concept introduced in the seminal Graph Minors project by Robertson and Seymour~\cite{RobertsonS91:GM10} and plays an important role in the proof of the graph minor theorem. 
This graph parameter measures how close a graph is to a tree, and is known to be linearly related to \emph{treewidth}~\cite{RobertsonS91:GM10}, which serves a similar purpose in solving problems on graphs. 
Since the branchwidth of a graph can be easily generalized to hypergraphs, matroids, and symmetric submodular functions, due to this versatility, its mathematical properties and algorithmic aspects have been extensively studied in the literature~\cite{FominT16,GeelenGRW06,JeongKO21,KorhonenO26,RobertsonS91:GM10}.

Algorithms for computing the treewidth and branchwidth of a graph have been well studied in the context of exact algorithms.
The main reason for this is that, given a decomposition of small width for a given graph, numerous optimization problems on graphs can be solved efficiently~\cite{ArnborgLS91,Bodlaender88,Courcelle90}, with running time typically linear in the size of an input graph and exponential in these width parameters.
This implies that it is crucial to exactly compute the treewidth or branchwidth of a graph.
Given a graph $G$ and an integer $k$, the problems of deciding whether the treewidth or branchwidth of $G$ is at most $k$ are known to be fixed-parameter tractable (FPT) parameterized by $k$, with the running time dependencies linear in the number of vertices~\cite{Bodlaender96,BodlaenderT97}.
However, since the dependence on $k$ in these running times is non-negligibly large, exact exponential-time algorithms are often preferred in practice.

For computing treewidth, Bouchitt\'e and Todinca~\cite{BouchitteT01} proposed a dynamic programming algorithm for computing treewidth based on potential maximal cliques, which leads to a running time bound essentially depending on the number of potential maximal cliques in an input graph.
The current best running time bound is due to Fomin, Todinca, and Villanger~\cite{FominTV15} and is in $\bigOh(1.7347^n)$, where $n$ is the number of vertices in the input graph. 
Interestingly, the essential idea behind this algorithm can be leveraged to develop highly practical algorithms; specifically, a framework called positive-instance driven dynamic programming~\cite{Tamaki19}, introduced by one of the current authors, successfully computes the treewidth of moderate-sized graphs.
Exact exponential-time algorithms are also known for branchwidth, including an $\bigOh^*(2^m)$-time algorithm\footnote{The notation $\bigOh^*$ hides a polynomial factor in the input size.} by Oum~\cite{Oum09} and an $\bigOh^*((2\sqrt{3})^n) \subset \bigOh(3.4652^n)$-time algorithm by Fomin, Mazoit, and Todinca~\cite{FominMT09}, where $m$ is the number of edges in the input graph.
However, these algorithms either compute branchwidth values across all edge subsets or involve highly complex procedures; consequently, a straightforward implementation of these methods does not appear to be practical. 
In fact, the current state-of-the-art practical algorithms for branchwidth are Hicks' algorithm~\cite{Hicks05}, which utilizes the notion of tangles, as well as the SAT-encoding approach proposed by Lodha, Ordyniak, and Szeider~\cite{LodhaOS19}.
Nevertheless, as reported in~\cite{LodhaOS19}, these existing approaches still cannot successfully compute the branchwidth of moderate-sized graphs on which Tamaki's algorithm can compute the treewidth exactly within a practical time limit.

In this paper, we aim to take a step forward in the practical aspects of branchwidth computation.
We present practical algorithms for computing the branchwidth of a (hyper)graph.
Our first algorithm is based on a relatively simple recurrence for computing branchwidth and runs in $\bigOh^*(4^n)$ time. 
The basic idea of this algorithm is to evaluate a recurrence based on the (rooted) tree structure of an optimal branch-decomposition using dynamic programming, similar to the approach taken by Oum~\cite{Oum09}.
However, a major improvement is that our running time is single-exponential with respect to $n$, the number of vertices in the input graph. 
The technical key ingredient to achieving this improvement lies in showing a recurrence over vertex subsets rather than edge subsets.
While this approach shares an underlying concept with the algorithm for computing linearwidth given by Kobayashi and Nakahata~\cite{KobayashiN20}, proving its correctness for branch-decompositions requires non-trivial extensions.
It is worth noting that this is the first algorithm for computing the branchwidth of hypergraphs that runs in time single-exponential in $n$; the algorithm given by Fomin, Mazoit, and Todinca~\cite{FominMT09} seems to work only on graphs. 

\begin{theorem}\label{thm:main-hg}
    The branchwidth of an $n$-vertex hypergraph can be computed in time $\bigOh^*(4^n)$.
\end{theorem}

Our second algorithm builds upon the recurrence of the first one to provide a faster algorithm specialized for graphs. This improves upon one of the best known theoretical running time bounds due to \cite{FominMT09}, proving the following theorem.

\begin{theorem}\label{thm:main}
    The branchwidth of an $n$-vertex graph can be computed in time $\bigOh(3.293^n)$.
\end{theorem}

This improvement is achieved by building our recurrence on vertex subsets satisfying a certain connectivity condition, which we call \emph{blocks}.
While this block-based algorithm improves the running time as stated in the theorem, the recurrence itself becomes more involved, and its evaluation becomes non-trivial.

We also give an alternative to the algorithm stated in \cref{thm:main}.
Based on some partial analysis, we conjecture
that it runs in time $\bigOh(c^n)$
where $c$ is a constant smaller than
the one in \cref{thm:main}.
In practice, it performs significantly better on instances
that are hard for the algorithm in 
\cref{thm:main}.


\paragraph*{Related work.}
Although computing the branchwidth of a graph is known to be NP-hard~\cite{SeymourT94} even when restricted to several specific graph classes~\cite{KloksKM05}, a highly notable result is that it can be computed in polynomial time for planar graphs~\cite{SeymourT94}. 
This is in contrast to computing treewidth for planar graphs, whose computational complexity has been a long-standing open problem.
Seymour and Thomas~\cite{SeymourT94} showed that the branchwidth of a planar graph can be computed in time $\bigOh(n^2)$. 
Their algorithm takes time $\bigOh(n^4)$ to construct the corresponding branch-decomposition, which was later improved to $\bigOh(n^3)$ by Gu and Tamaki~\cite{GuT08}.
The computation of branchwidth for planar graphs has also been extensively investigated from a practical perspective, enabling the exact computation of optimal branch-decompositions within a practical running time, even for large planar graphs~\cite{BianGZ16,Hicks05,Hicks05a}.

\section{Preliminaries}

Let $G$ be a hypergraph.
The vertex set and edge set of $G$ are denoted by $V(G)$ and $E(G)$, respectively.
For $F \subseteq E(G)$, we define the set of vertices contained in some (hyper)edge of $F$ as $V(F) = \bigcup_{e \in F} e$.
For $X \subseteq V(G)$, the \emph{boundary} of $X$, denoted by $\bound_G(X)$, is the set of vertices in $X$ that have an incident hyperedge $e$ with $e \cap (V(G)\setminus X) \neq \emptyset$.
We can extend this definition for edge sets: For $F \subseteq E(G)$, the boundary of $F$, denoted by $\bound_G(F)$ is defined to be $V(F) \cap V(E(G) \setminus F)$.
For $X \subseteq V(G)$, we define $E_G(X) = \{e \in E(G) : e \subseteq X\}$, and the subgraph of $G$ induced by $X$ is denoted by $G[X]$, i.e., $V(G[X]) = X$ and $E(G[X]) = E_G(X)$.
We write $\closure(F)$ to denote the set of edges in the subgraph induced by $V(F)$, that is, $\closure(F) = E(V(F))$ for $F \subseteq E(G)$.
An edge set $F \subseteq E(G)$ is \emph{closed} if it satisfies $\closure(F) = F$.
In other words, $F$ is closed if and only if there is a vertex set $X$ such that $E(X) = F$.
If $G$ is a graph, the sets of open and closed neighbors of $X \subseteq V(G)$ are denoted by $N_G(X)$ and $N_G[X]$, respectively.
When the hypergraph $G$ is clear from the context, we may omit it from these notations. 

Let $T$ be a rooted binary tree.
The set of leaves in $T$ is denoted by $L(T)$.
We particularly refer to the vertices of $T$ as nodes to distinguish them from the vertices of $G$.
For a node $p \in V(T)$, the subtree of $T$ rooted at $p$ is denoted by $T_p$.
We denote by $T - T_p$ the rooted tree obtained from $T$ by deleting all nodes in $T_p$.
We say that $T$ is \emph{full} if all internal nodes have exactly two children.

A \emph{branch-decomposition} of a hypergraph $G$ is a pair consisting of a rooted full binary tree $T$ and a bijection from the leaves of $T$ to $E(G)$.
Let $\calT = (T, \varphi)$ be a branch-decomposition of $G$.
We may refer to the nodes and edges of $T$ as the nodes and edges of $\calT$ as well.
For a subtree $T'$ of $T$, we write $\varphi(T')$ to denote the set of edges to which $\varphi$ maps the leaves of $T$ belonging to $T'$.
For a node $p \in V(T)$, the \emph{middle set} of $p$ is defined as $\midset_{\calT}(p) = V(\varphi(T_p)) \cap V(\varphi(T - T_p))$, that is, the set of vertices in $G$, each of which has at least one incident edge in both $\varphi(T_p)$ and $\varphi(T - T_p)$.
The \emph{order} of a node $p$ in $\calT$ is defined as $|\midset_{\calT}(p)|$.
The \emph{width} of $\calT$ is defined as the maximum order of a node in $\calT$, and the \emph{branchwidth} of $G$, denoted by $\bw(G)$, is the minimum integer $k$ such that $G$ has a branch-decomposition of width at most $k$.
A branch-decomposition can be defined for edge subsets: A \emph{branch-decomposition} of $F \subseteq E(G)$ is a pair $\calT = (T, \varphi)$, where $T$ is a rooted full binary tree and $\varphi$ is a bijection from the leaves of $T$ to $F$.
For each node $p$ of $\calT$, the \emph{lower edge set} of $p$ in $\calT$, 
denoted by $\lowset_{\calT}(p)$, is defined as $\lowset_{\calT}(p) = \varphi(T_p)$;
the \emph{upper edge set} of $p$ in $\calT$, denoted by $\upset_{\calT}(p)$, 
is defined as $\upset_{\calT}(p) = \varphi(T - T_p) \cup (E(G) \setminus F)$. 
Note that these definitions are asymmetric in the sense that the upper edge set contains all the edges outside $F$.
The \emph{middle set} of $p$ in $\calT$, denoted by $\midset_{\calT}(p)$, is 
defined as $\midset_{\calT}(p) = V(\lowset_{\calT}(p)) \cap V(\upset_{\calT}(p))$.
The order of a node, the width of $\cal T$, and the branchwidth of $F$ of $G$ are defined analogously.
We denote the branchwidth of $F$ of $G$ by $\bw_G(F)$.
The notions of the upper set, the middle set, and the order of a node of a branch-decomposition
depend on the hypergraph $G$ under consideration. We, however, do not make an explicit reference
to $G$ when using those notions and related notation unless there is a fear of confusion.
It is easy to see that $\bw(G) = \bw(E(G))$.

We use a relaxed version of branch-decompositions~\cite{RobertsonS91:GM10}.
A pair $\mathcal T = (T, \varphi)$ is called a \emph{relaxed branch-decomposition} of $G$ if $T$ is a rooted (not necessarily full) binary tree and $\varphi$ is a \emph{surjection} from $L(T)$ to $E(G)$, that is, there can be two distinct leaves $\ell$ and $\ell'$ with $\varphi(\ell) = \varphi(\ell')$.
The order of a node (and hence the width of $\mathcal T$) is defined in an analogous way.
For every graph $G$, the minimum width of a relaxed branch-decomposition of $G$ is equal to $\bw(G)$~\cite{RobertsonS91:GM10}.
In this paper, we simply refer to a relaxed branch-decomposition $\calT=(T, \varphi)$ as a branch decomposition; $\calT$ is \emph{full} if $T$ is a full binary tree and is \emph{injective} if $\varphi$ is an injection.
The following is a straightforward generalization of standard observations on relaxed branch-decompositions.
\begin{proposition}
    \label{prop:bw_injective}
    Let $F$ be an edge set of a hypergraph $G$ of branchwidth $k$.
    Then, there is an injective branch-decomposition of $F$ of width $k$,
    there is a full branch-decomposition of $F$ of width $k$, and
    there is a full and injective branch-decomposition of $F$ of width $k$.    
\end{proposition}

In the subsequent sections, we will assume that the input hypergraph is connected, has no vertex belonging to at most one hyperedge, and has no two hyperedges such that one contains the other.
It is easy to verify that the branchwidth of a hypergraph is equal to the maximum branchwidth of its connected components.
The other two assumptions are justified by the following two easy propositions.

\begin{proposition}\label{prop:private-vertex}
    Let $G$ be a hypergraph and let $v$ be a vertex that belongs to at most one hyperedge in $G$.
    Let $G'$ be the hypergraph with $V(G') = V(G) \setminus \{v\}$ and $E(G') = \{e \setminus \{v\} : e \in E(G)\}$. 
    Then, $\bw(G) = \bw(G')$.
\end{proposition}
\begin{proof}
    Let $\calT = (T, \varphi)$ be an injective branch-decomposition of $G$ and let $\calT' = (T, \varphi')$ be a branch-decomposition of $G'$, where $\varphi'$ maps each leaf $\ell$ of $T$ to the edge $\varphi(\ell) \setminus \{v\}$. 
    Clearly, the width of $\calT'$ is at most the width of $\calT$.
    From an injective branch-decomposition $\calT'$ of $G'$, we can define a branch-decomposition $\calT$ of $G$ in the opposite way.
    As $v$ does not appear in any middle set of a node in $\calT$, the width of $\calT$ is at most the width of $\calT'$.
\end{proof}

\begin{proposition}\label{prop:sperner}
    Let $G$ be a hypergraph.
    Suppose that $G$ has a hyperedge $e$ and a distinct hyperedge $e'$ with $e \subseteq e'$.
    Let $G'$ be the graph obtained from $G$ by deleting $e$.
    Then, $\bw(G) = \max\{\bw(G'), |e|\}$.
\end{proposition}
\begin{proof}
    Since $G'$ is a sub-hypergraph of $G$, we have $\bw(G') \le \bw(G)$.
    Moreover, for every branch-decomposition $\mathcal T = (T, \varphi)$, the middle set of a leaf node $\ell$ with $\varphi(\ell) = e$ contains all vertices in $e$.
    Thus, we have $|e| \le \bw(G)$.

    Conversely, let $\mathcal T' = (T', \varphi')$ be an injective branch-decomposition of $G'$.
    Let $\ell'$ be a leaf of $T'$ such that $\varphi'(\ell') = e'$ and let $p$ be the parent of $\ell'$.
    Then, we subdivide the edge between $p$ and $\ell'$ with a new node $q$ and add a new leaf $\ell$ as a child of $q$.
    We denote by $T$ the tree obtained as above.
    We also extend the mapping $\varphi'$ by mapping $\ell$ to $e$ and denote the extended mapping by $\varphi$.
    Then, $\calT = (T, \varphi)$ is an injective branch-decomposition of $G$.
    As $e \subseteq e'$, the order of $\ell$ equals $|e|$, the order of $q$ equals the order of $p$, and the order of every other node in $\cal T$ equals its order in $\calT'$. 
    Hence, $\bw(G) \le \max\{\bw(G'), |e|\}$.
\end{proof}

\section{An $\bigOh^*(4^n)$-time algorithm for hypergraphs}
In this section, we develop a branchwidth algorithm for hypergraphs that runs in 
time $\bigOh^*(4^n)$.

We assume the hypergraph $G$ for which we compute the branchwidth is connected,  
every vertex of $G$ belongs to at least two edges of $G$, and no edge of $G$ is
strictly contained in another edge of $G$. 
By \cref{prop:private-vertex,prop:sperner}, we can reduce the branchwidth computation of general hypergraphs to the branchwidth computation of hypergraphs
that satisfy these conditions.

The following observations are crucial in reasoning about the branchwidth of edge sets.

\begin{lemma}
    \label{prop:bd-bound}
   Let $\calT$ be an injective branch-decomposition of an edge set of $G$ and
   let $p_1$ and $p_2$ be two nodes of $\calT$ that share a common parent.
   Let $L_i = \lowset_{\calT}(p_i)$ for $i = 1, 2$.
   Then, $V(L_1) \cap V(L_2) = \bound(L_1) \cap \bound(L_2)$.
\end{lemma}
\begin{proof}
Since $\calT$ is injective, $L_1$ and $L_2$ are disjoint from each other.
Let $v$ be a vertex in $V(L_1) \cap V(L_2)$. Since $v$ is incident 
to an edge in $L_1$ as well as to an edge in $L_2$, 
it belongs to $\bound(L_1)$ and to $\bound(L_2)$.
Therefore, $V(L_1) \cap V(L_2)$ is contained in $\bound(L_1) \cap \bound(L_2)$.
The inclusion in the other direction is trivial as $\bound(L_i) \subseteq V(L_i)$ for $i = 1, 2$.
\end{proof}

Let $\calT = (T, \varphi)$ be a branch-decomposition of an edge set of a hypergraph $G$ and $p$ a node of $T$.
The \emph{sub-branch-decomposition} of $\calT$ at $p$ is the branch-decomposition
$(T_p, \varphi_p)$ where $T_p$ is the subtree of $T$ rooted at $p$
and $\varphi_p$ is the restriction of $\varphi$ to the leaves that belong to $T_p$.
\begin{lemma}
    \label{prop:sub-bd}
    Let $\calT$ be an injective branch-decomposition of an edge set of $G$ and $p$ a node of $\calT$. Then, the width of the sub-branch-decomposition $\calT_p$ of $\calT$
    at $p$ is at most the width of $\calT$.
\end{lemma}
\begin{proof}
    Let $q$ be a node of $\calT_p$. The lower sets of $q$ in $\calT_p$ and in $\calT$
    are identical. Since $\calT$ is injective, the upper sets of $q$ in $\calT_p$ and
    in $\calT$ are also identical.
    Therefore, the order of $q$ in $\calT_p$ is equal to its order in $\calT$.
\end{proof}

\begin{lemma}
    \label{prop:bd-edge-deletion}
    Let $F_1$ and $F_2$ be edge sets of $G$ such that $F_1 \subseteq F_2$ 
    and each edge $e \in F_2 \setminus F_1$ intersects $V(F_1)$ only in $\bound(F_2)$, that is, $e \cap V(F_1) \subseteq \bound(F_2)$.
    Then, we have $\bw(F_1) \leq \bw(F_2)$.
\end{lemma}
\begin{proof}
    Let $\calT$ be a branch-decomposition of $F_2$ of width $\bw(F_2)$.
    For each $e \in F_2 \setminus F_1$, we remove the leaves that are mapped to $e$, and then repeatedly remove the new leaves that are internal nodes in $\calT$.
    Then, let $\calT'$ be the branch-decomposition of $F_1$ obtained in this way.
    For each node $p$ of $\calT'$, it holds that $\lowset_{\calT'}(p) = \lowset_{\calT}(p) \cap F_1$ and 
    $\upset_{\calT'}(p) = \upset_{\calT}(p) \cup (F_2 \setminus F_1)$.
    Moreover, we have $V(\upset_{\calT'}(p)) \subseteq V(\upset_{\calT}(p))$ as every edge $e \in F_2 \setminus F_1$ intersects $V(F_1)$ only in $\bound(F_2)$ by assumption, and $\bound(F_2) \subseteq V(\upset_{\calT}(p))$ as $\upset_{\calT}(p)$ contains $E(G) \setminus F_2$.
    Hence, $\midset_{\calT'}(p) \subseteq \midset_{\calT}(p)$, implying that the order of $p$ in $\calT'$ is at most that in $\calT$.
\end{proof}

Let $\calT_1 = (T_1, \varphi_1)$ and $\calT_2 = (T_2, \varphi_2)$ be branch-decompositions 
of edge sets of a hypergraph $G$.
The \emph{composition} of $\calT_1$ and $\calT_2$ is the branch-decomposition
$\calT = (T, \varphi)$ defined as follows.
The tree $T$ has a root $r$ with two child nodes identified with the roots of $T_1$ and $T_2$.
The mapping $\varphi$ maps a leaf $\ell$ to $\varphi_i(\ell)$, where $i$ is such that $\ell$ belongs to
$T_i$.

\begin{lemma}
    \label{prop:bd-compose}
    Let $F_1$ and $F_2$ be edge sets of $G$ such that
    $V(F_1) \cap V(F_2) = \bound(F_1) \cap \bound(F_2)$.
    Let $\calT_i$ be a branch-decomposition of $F_i$, for $i = 1, 2$ and
    let $\calT$ be the composition of $\calT_1$ and $\calT_2$. 
    Then the width of $\calT$ is at most the largest of the width of $\calT_1$,
    the width of $\calT_2$, and $|\bound(F)|$, where $F = F_1 \cup F_2$.
\end{lemma}
\begin{proof}
Let $p$ be a node of $\calT$ that belongs to $\calT_1$. 
Then, we have $\lowset_{\calT}(p) = \lowset_{\calT_1}(p)$ and $\upset_{\calT}(p) = \upset_{\calT_1}(p) \cup F_2$.
Since each edge $e \in F_2$ intersects $V(F_1)$
only in $\bound(F_1)$ and $\bound(F_1) \subseteq V(\upset_{\calT_1}(p))$,
we have $\midset_{\calT}(p) \subseteq \midset_{\calT_1}(p)$.
Thus, the order of each node $p$ of $\calT$ that belongs to $\calT_1$
is at most the width of $\calT_1$. Similarly, the order of each node $p$ of $\calT$ that belongs 
to $\calT_2$ is at most the width of $\calT_2$.
The middle set of the root $r$ of $\calT$ is $V(F) \cap V(E(G) \setminus F)) = \bound(F)$
and hence its order is equal to $|\bound(F)|$.
\end{proof}

\begin{lemma}
\label{lem:bw_egdge_closure}
For $F \subseteq E(G)$, we have $\bw(\closure(F)) \leq \bw(F)$.
\end{lemma}
\begin{proof}
    Suppose that $\bw(F) = k$ and let $\calT = (T, \varphi)$ be a branch-decomposition
    of $F$ of width $k$. Let $F^+ = \closure(F) \setminus F$.
    Observe that each edge $e \in F^+$ is contained in $\bound(F) = V(F) \cap V(E(G) \setminus F)$.
    This follows since $V(F) = V(\closure(F))$ and $e \in E(G) \setminus F$.
    Thus, we have $V(F^+) \subseteq \bound(F)$.
    Let $T^+$ be an arbitrary branch-decomposition of $F^+$. 
    Since $|V(F)| \leq |\bound(F)| \leq \bw(F) \leq k$,
    the width of $T^+$ is at most $k$. Let $\calT'$ be the composition of $\calT$ and $\calT^+$.
    Then, by \cref{prop:bd-compose}, the width of $\calT'$ is the largest of the width of $\calT$,
    the width of $\calT^+$, and $|\bound(\closure(F))|$. 
    Since $V(F) = V(\closure(F))$ and
    $E(G) \setminus \closure(F) \subseteq E(G) \setminus F$, $\bound(\closure(F)) \subseteq \bound(F)$. 
    Hence, the width of $\calT'$ is at most $k$.
\end{proof}
For vertex sets $X_1$ and $X_2$ of $G$, 
we define $F^+(X_1, X_2) = E(X_1 \cup X_2) \setminus (E(X_1) \cup E(X_2))$.

\begin{lemma}
    \label{prop:bound-union}
    Let $X_1$ and $X_2$ be vertex sets of a hypergraph $G$.
    Then, $\bound(E(X_1) \cup E(X_2))$ is contained in 
    $\bound(X_1 \cup X_2) \cup V(F^+(X_1, X_2))$.
\end{lemma}
\begin{proof}
Let $X = X_1 \cup X_2$.
Let $e$ be an arbitrary edge in $E(G) \setminus (E(X_1) \cup E(X_2))$. 
Suppose that $e \cap (V(G) \setminus X) \neq \emptyset$.
Then, $e$ intersects $V(E(X_1) \cup E(X_2)) = X$ only in $\bound(X)$.
Suppose otherwise.
Then, $e$ belongs to $F^+(X_1, X_2)$ and hence is contained in $V(F^+(X_1, X_2))$.
Therefore, $\bound(E(X_1) \cup E(X_2))$ is contained in $\bound(X) \cup V(F^+(X_1, X_2))$.
\end{proof}

We are now ready to prove the recurrence of our algorithm.

\begin{lemma}
   \label{lem:bw-recurrence} 
   Let $X$ be a vertex set of $G$ with $|E(X)| \geq 2$. Then, $\bw(E(X)) \leq k$ if and 
   only if there are proper subsets $X_1$ and $X_2$ of $X$ such that
   \begin{enumerate}
   \item $X_1 \cup X_2 = X$, 
       \item $\bw(E(X_i)) \leq k$ for $i = 1, 2$, 
       \item $X_1 \cap X_2 = \bound(X_1) \cap \bound(X_2)$, and
       \item $|\bound(X) \cup V(F^+(X_1, X_2))| \leq k$.
   \end{enumerate}
\end{lemma}
\begin{proof}
    Suppose first that there are proper subsets $X_1$ and $X_2$ of $X$ such that the four conditions in the lemma
    are satisfied. Let $\calT_i = (T_i, \varphi_i)$ be a full and injective branch-decomposition
    of $E(X_i)$ for $i = 1, 2$. Let $\calT = (T, \varphi)$ be the composition of
    $\calT_1$ and $\calT_2$. Then, $\calT$ is a branch-decomposition of $E(X_1) \cup E(X_2)$
    and, by \cref{prop:bd-compose}, its width is the largest of the width of $\calT_1$,
    the width of $\calT_2$, and $|\bound(E(X_1) \cup E(X_2))|$.
    The first two are at most $k$ by the second condition in the assumption.
    Since $\bound(E(X_1) \cup E(X_2))$ is a subset of 
    $\bound(X) \cup V(F^+(X_1, X_2))$ due to \cref{prop:bound-union},
    the last is also at most $k$ by the fourth condition in the assumption.
    We conclude that the width of $\calT$ is at most $k$ and hence
    $\bw(E(X_1) \cup E(X_2)) \leq k$. Since $\closure(E(X_1) \cup E(X_2)) = E(X)$, 
    we may apply \cref{lem:bw_egdge_closure} and conclude that 
    $\bw(E(X))$ is at most $k$ as well.

    For the converse, suppose that $X$ is a vertex set of $G$ with $|E(X)| \geq 2$ and $\bw(E(X)) = k$.     
    We prove that there are proper subsets $X_1$ and $X_2$ of $X$ that satisfy the four conditions
    in the lemma.
    Let $\calT = (T, X)$ be a branch-decomposition of $E(X)$ of width $k$.
    We can assume that $\calT$ is full and injective due to \cref{prop:bw_injective}.
    Let $p$ be the furthest node of $T$ from the root such that $\lowset_\calT(p) = X$.
    Since $|E(X)| \geq 2$ and no edge of $G$ strictly contains 
    another edge by our assumption on $G$, $p$ is not a leaf.
    Moreover, as $p$ is the furthest node from the root such that $\lowset_\calT(p) = X$, $p$ has two child nodes $p_1$ and $p_2$ such that $V(\lowset_{\calT}(p_i))$ is a proper subset of $X$ for $i = 1, 2$. 
    We let $X_i = V(L_{\calT}(p_i))$ for $i = 1, 2$.
    The first condition that $X_1 \cup X_2 = X$ holds since $L_{\calT}(p_1) \cup L_{\calT}(p_2) = 
    L_{\calT}(p)$. 
    Since $\calT$ is injective, we have $X_1 \cap X_2 = \bound(X_1) \cap \bound(X_2)$ by \cref{prop:bd-bound}: the third condition holds.

    It remains to show that the second and last conditions hold.
    Let $L_i = L_{\calT}(p_i)$ for $i = 1, 2$, and let $\calT_i = (T_i, \varphi_i)$
    be the sub-branch-decomposition of $\calT$ at $p_i$, which is a branch-decomposition of $L_i$.
    Since $\calT$ is injective, the width of $\calT_i$ is at most $k$ for $i = 1, 2$,
    by \cref{prop:sub-bd} and hence $\bw(L_i)$ is at most $k$.
    Since $E(X_i) = \closure(L_i)$, 
    $\bw(E(X_i))$ is at most $k$ for $i = 1, 2$ by \cref{lem:bw_egdge_closure}:
    the second condition in the lemma holds.
    For the last condition, observe that $V(\lowset_{\calT}(p)) = X_1 \cup X_2 = X$ and
    $E(G) \setminus (E(X_1) \cup E(X_2)) \subseteq E(G) \setminus \lowset_{\calT}(p)$.
    Therefore, both $\bound(X)$ and $V(F^+(X_1, X_2))$ are subsets of
    $\bound(\lowset_{\calT}(p))$. As $\bound(\lowset_{\calT}(p))$ is contained in
    the middle set of $p$, its cardinality is at most $k$.
    Therefore, we have $|\bound(X) \cup V(F^+(X_1, X_2))| \leq |\bound(L_{\calT}(p))| \leq k$.
\end{proof}

Our algorithm simply evaluates the recurrence in \cref{lem:bw-recurrence} in a dynamic programming manner.
Fix an integer $k \ge 0$.
We can decide if $\bw(E(X)) \le k$ for all $X \subseteq V(G)$ in total time $\bigOh^*(4^n)$, where $n = |V(G)|$.
Since the branchwidth of $G$ is the smallest integer $k$ such that $\bw(E(G)) \le k$ holds, this completes the proof of~\cref{thm:main-hg}.

\section{An $\bigOh(3.293^n)$-time algorithm for graphs}
In this section, we improve the running time of the previous algorithm for \emph{graphs}.

\subsection{Blocks and block derivations}
The \emph{core} of $X$, denoted by $\core_G(X)$ or simply by $\core(X)$, is defined to be 
$\core(X) = X \setminus \bound(X)$, that is, the subset of vertices in $X$ that has no neighbors outside $X$.
The \emph{core} of an edge set $F$, denoted by $\core_G(F)$ or simply $\core(F)$, is defined as $\core(F) = V(F) \setminus \bound(F)$. 
Note that $\bound(F) = \bound(V(F))$ if $F$ is closed.

For brevity, we refer to the closed neighborhood of a vertex set $W$ as 
the \emph{vertex closure} or simply the \emph{closure} of $W$ and
use the notation $\hat{W}$ interchangeably with $N[W]$ when $G$ is clear from the context.
We denote by $\corehat(X)$ the closure of $\core(X)$.

We call a vertex set $X$ of a hypergraph $G$ a \emph{block} 
if $\core(X)$ is non-empty and connected and, moreover,
$X = \corehat(X)$. Under our assumption that $G$ is connected, $V(G)$ is a block
with $\corehat(V(G)) = V(G)$. 
It should be noted that while the term ``block'' is also used in the paper by Fomin, Mazoit, and Todinca~\cite{FominMT09}, our definition, however, differs from theirs.

Let $B$ be a block of a graph $G$. A \emph{block derivation} of $B$ in $G$ is 
a pair $D = (B, \calB)$ where $\calB$ is a set of blocks such that
$A \subseteq B$ for each $A \in \calB$ and 
$\core(A_1) \cap \core(A_2) = \emptyset$ for each pair of distinct $A_1, A_2 \in \calB$.
For a block derivation $D = (B, \calB)$, we use the following notation:
$B(D) \coloneqq B$, $\calB(D) \coloneqq \calB$, $K(D) \coloneqq \bigcup_{A \in \calB} \core(A)$, and
$S(D) \coloneqq B \setminus K(D)$. 


A triple $(M_1, M_2, M_3)$ of subsets of $S(D)$ is a \emph{mid-triple} of a block derivation $D$
if the following conditions are satisfied:
\begin{enumerate}
    \item $M_1 \cup M_2 = M_2 \cup M_3  = M_3 \cup M_1 = S(D)$,
    \item $\bound(B) \subseteq M_3$, and
    \item for each $A \in \calB(D)$, $\bound(A)$ is contained in either $M_1$ or $M_2$.
\end{enumerate}
The \emph{order} of a mid-triple $(M_1, M_2, M_3)$ is the largest of $|M_1|$, $|M_2|$, and $|M_3|$.
The \emph{order} of a block derivation $D$ is the smallest $k$ such that
there is a mid-triple of $D$ of order $k$.
The \emph{width} of a block derivation $D$ is the larger of the order of $D$
and the largest of $\bw(E(A))$ over all $A \in \calB(D)$.

\begin{lemma}
    \label{lem:block-deriv2bw}
    Let $B$ be a block of a graph $G$. Suppose that there is a block derivation 
    $D$ of $B$ of width~$k$. 
    Then, $\bw(E(B))$ is at most $k$.
\end{lemma}
\begin{proof}
Let $(M_1, M_2, M_3)$ be a mid-triple of $D$ of order at most $k$.
Let $(\calB_1, \calB_2)$ be a bipartition of $\calB(D)$ such
that $\bound(A) \subseteq M_i$ for every $A \in \calB_i$, for $i = 1, 2$.

We first define several subsets of $E(B)$.
For each $A \in \calB(D)$, let $F_A = E(A)$.
We define $F_i$ as $F_i = \bigcup_{A \in \calB_i} F_A$, for $i = 1, 2$.
We denote by $F_{i + 2}$ the set of edges of $G$ contained in $M_i$ for $i = 1, 2, 3$, that is, $F_{i + 2} \coloneqq E(M_i)$.
We claim that $E(B) = F_1 \cup F_2 \cup F_3 \cup F_4 \cup F_5$.
For each $1 \le i \le 5$, we have $V(F_i) \subseteq B$ and
hence $F_i \subseteq E(B)$. Therefore, $F_1 \cup F_2 \cup F_3 \cup F_4 \cup F_5$ is
a subset of $E(B)$. For the inclusion in the other direction, let 
$e$ be an edge in $E(B)$, and we show that $e$ belongs to some $F_i$.
If $e$ is contained in some $A \in \calB(D)$,
then clearly $e$ belongs to $F_1$ or $F_2$. 
Suppose otherwise.
Since $e \in E(B) \setminus (F_1 \cup F_2)$ intersects $B \setminus V(F_i)$ for each $i = 1,2$, it has no intersection with $\core(A)$ for any $A \in \calB(D)$.
This implies that $e \subseteq S(D)$.
Since every vertex in $S(D)$ belongs to at least two of $M_1$, $M_2$, and
$M_3$, and moreover \emph{$G$ is a graph}, there is some $i \in \{1, 2, 3\}$ such that $e \subseteq M_i$ and hence $e$ belongs to $F_{i + 2}$.

We construct a branch-decomposition $\calT$ of $E(B)$, which is
composed of several sub-branch-decompositions described below.
For each $A \in \calB(D)$, let $\calT_A = (T_A, \varphi_A)$ be
a branch-decomposition of $F_A$ of width at most $k$.
We can assume that $\calT_A$ is injective due to \cref{prop:bw_injective}.
For $i = 1, 2$, let $m_i = |\calB_i|$ and let $A^i_1$, \ldots, $A^i_{m_i}$ be 
a listing of the members of $\calB_i$ in an arbitrary order.
For $i = 1, 2$ and for $1 \leq j \leq m_i$, let $\calT^i_j$
be a branch-decomposition of $\bigcup_{1 \leq h \leq j}F_{A^i_h}$ inductively defined as follows:
$\calT^i_1 = \calT_{A^i_1}$ and $\calT^i_j$ for $2 \leq j \leq m_i$ is the
composition of $\calT^i_{j - 1}$ and $\calT_{A^i_j}$.
We denote $\calT^i_{m_i}$ by $\calT_i$, for $i = 1, 2$.
For $i = 3, 4, 5$, $\calT_{i}$ is an arbitrary branch-decomposition of $F_i$. 
The composition of $\calT_1$ and $\calT_3$ is denoted by $\calT_6$ and
the composition of $\calT_2$ and $\calT_4$ is denoted by $\calT_7$. 
Then, $\calT_8$ is the composition of $\calT_6$ and $\calT_7$.
Finally, $\calT$ is the composition of $\calT_8$ and $\calT_5$.
To summarize, $\calT_i$ is a branch-decomposition of $F_i$ for $1 \le i \le 5$,
$\calT_6$ is a branch-decomposition of $F_1 \cup F_3$, 
$\calT_7$ is a branch-decomposition of $F_2 \cup F_4$,
$\calT_8$ is a branch-decomposition of $F_1 \cup F_2 \cup F_3 \cup F_4$, and
$\calT$ is a branch-decomposition of $F_1 \cup F_2 \cup F_3 \cup F_4 \cup F_5 = E(B)$.
We now claim that the width of $\calT$ is at most $k$.
To see this, let $p$ be a node of $\calT$.

Suppose first that $p$ belongs to $\calT_A$ for some $A \in \calB(D)$. 
Then the lower edge set of $p$ in $\calT$ is identical to that in $\calT_A$. 
Let $e \in \upset_{\calT}(p) \setminus \upset_{\calT_A}(p)$.
Since $\calT_A$ is injective and $e \notin \upset_{\calT_A}(p)$, it holds that $e \in \lowset_{\calT_A}(p)$, meaning that $e$ belongs to $F_A$. 
Moreover, as there is a leaf of $\calT$ not in $\calT_A$ that is mapped to $e$, $e$ is not incident
to a vertex in $\core(A)$. 
Hence, we have $e \subseteq \bound(A)$, which is a subset of $V(\upset_{\calT_A}(p))$, and hence
the middle set of $p$ in $\calT$ is identical to its middle set in $\calT_A$. Therefore, 
the order of $p$ in $\calT$ is at most $k$.

Suppose $p$ is the root of $\calT^1_j$ for some $1 \le j \le m_1$. 
Then, $V(\lowset_{\calT}(p))$ is a subset of $\bigcup_{1 \leq h \leq j} A^1_h$. Since
no leaf of $\calT$ that is not a descendant of $p$ is mapped to an edge that intersects  $\bigcup_{1 \leq h \leq j} \core(A^1_h)$,
the middle set of $p$ is a subset of $M_1$ and hence its order is at most~$k$. The case where $p$
is the root of $\calT^2_j$ for some $1 \le j \le  m_2$ is similar. We are done with the nodes in
$\calT_1$ and $\calT_2$.

Suppose $p$ is a node in $\calT_{i}$ for some $3 \le i \le 5$. Then,
the lower edge set of $p$ is a subset of $M_{i - 2}$ and hence the order of $p$ is at most $k$.

Suppose $p$ is the root of $\calT_6$, the composition of $\calT_1$ and $\calT_3$.
Its lower edge set is $F_1 \cup F_3$ and its upper edge set is $E(G) \setminus (F_1 \cup F_3)$.
Therefore, its middle set is $\bound(F_1 \cup F_3) \subseteq M_1$ and hence its order is at most $k$.
Similarly, the order of the root of $\calT_7$ is at most $k$.

Suppose finally that $p$ is the root of $\calT_8$, the composition of $\calT_6$ and $\calT_7$. 
Its lower edge set is $F_1 \cup F_2 \cup F_3 \cup F_4$ and its upper edge set is
$E(G) \setminus (F_1 \cup F_2 \cup F_3 \cup F_4)$. Therefore, its middle set
is $\bound(F_1 \cup F_2 \cup F_3 \cup F_4)$. Let $e$ be an edge
in $E(G) \setminus (F_1 \cup F_2 \cup F_3 \cup F_4)$ that intersects $V(F_1 \cup F_2 \cup F_3 \cup F_4) = B$.
If $e$ intersects $V(G) \setminus B$ then its intersection with $B$ is contained in $\bound(B) \subseteq M_3$.
Otherwise, $e$ intersects both $M_1 \setminus M_2$ and $M_2 \setminus M_1$ and hence contained in $M_3$.
Therefore, the order of $p$ is at most $|M_3| \leq k$.

Finally, the middle set of the root of $\calT$ corresponds to $\bound(B)$ and hence its order is at most $k$.
We conclude that the width of $\calT$ is at most $k$ and hence $\bw(E(B))$ is at most $k$.
\end{proof}

\begin{theorem}
    \label{thm:bd-iff}
    Let $G$ be a graph and $B$ a block of $G$.
    Then, $\bw(E(B)) = k$ if and only if there is a block derivation $D$ of $B$ of width~$k$.
\end{theorem}
\begin{proof}
By \cref{lem:block-deriv2bw}, if there is a block derivation of $B$ of width~$k$ then
$\bw(E(B)) \leq k$. Therefore, to prove the theorem, it suffices to show that
if $\bw(E(B)) = k$ then there is a block derivation of $G$ of width at most $k$.
Suppose that $\bw(E(B)) = k$.
Let $\calT$ be a full and injective branch-decomposition of $E(B)$ of width $k$.
Since $B$ is a block and each vertex in $\core(B)$ is incident to at least two edges, 
$|E(B)| \geq 2$ and hence $\calT$ has at least two leaves. 
Therefore, similarly to the proof of \cref{lem:bw-recurrence}, there is a node $p$ of $T$
with child nodes $p_1$ and $p_2$ such that $V(\lowset_{\calT}(p)) = B$ but
$V(\lowset_{\calT}(p_i))$ is a proper subset of $B$, for $i = 1, 2$.
Let $L_i = \lowset_{\calT}(p_i)$, $X_i = V(L_i)$, and $M_i = \bound(L_i)$ for $i = 1, 2$.
Let $\calC$ be the set of components of $G[B \setminus (M_1 \cup M_2)]$ and
let $\calB = \{\hat{C} \mid C \in \calC\}$. 
Note that each $A \in \calB$ is contained in either $X_1$ or $X_2$.
Then, $D = (B, \calB)$ is a block derivation of $B$,
since each $A \in \calB$ is contained in $B$ and, for each pair of distinct members $A_1, A_2 \in \calB$,
$\core(A_1) \cap \core(A_2) = \emptyset$. 
We claim $(M_1, M_2, M_3)$, where $M_3 = \bound(B) \cup (M_1 \setminus M_2) \cup (M_2 \setminus M_1)$,
is a mid-triple of $D$. To see this, let $A$ be a member of $\calB$. 
By~\cref{prop:bd-bound}, it holds that $X_1 \cap X_2 = \bound(X_1) \cap \bound(X_2)$.
As $A$ is contained in either $X_1$ or $X_2$, $\bound(A)$ is contained in either $M_1$ or $M_2$.
Moreover, the definition of $M_3$ ensures that $M_1 \cup M_2 = M_2 \cup M_3 = M_1 \cup M_3 = S(D)$ and
that $\bound(B) \subseteq M_3$ as required by the definition of mid-triples. 

For $i = 1, 2$, since $|M_i| = |\bound(L_i)|$ is at most the order of the node $p_i$ in $\calT$, it is at most $k$.
To bound the cardinality of $M_3$, 
we claim that $M_3 \subseteq \bound(L_1 \cup L_2)$.
Every vertex in $\bound(B)$ belongs to $\bound(L_1 \cup L_2)$ since it
belongs to $V(L_1 \cup L_2) = B$ and is adjacent
to some vertex in $V(G) \setminus B$.
This implies that $\bound(B) \subseteq \bound(L_1 \cup L_2)$.
Let $v$ be a vertex in $(M_1 \setminus M_2) \setminus \bound(B)$.
Since $v$ belongs to $M_1 = \bound(L_1)$ but not to $\bound(B)$, 
there is an edge $e$ of $G$ between $v$ and some vertex $u$ in $X_2 \setminus X_1$.
As $v \notin M_2$, $u$ must belong to $\bound(L_2) \setminus X_1 =
M_2 \setminus M_1$. Thus, we have $e \in E(G) \setminus (L_1 \cup L_2)$.
As $e \subseteq V(L_1 \cup L_2)$, it holds that $v \in \bound(L_1 \cup L_2)$.
Hence, $M_1 \setminus M_2$ is a subset of $\bound(L_1 \cup L_2)$.
Similarly, $M_2 \setminus M_1$ is a subset of $\bound(L_1 \cup L_2)$.
Since $|\bound(L_1 \cup L_2)|$ is at most the order of $p$ in $G$ and hence is at most $k$,
$|M_3|$ is at most $k$ and hence the order of the mid-triple $(M_1, M_2, M_3)$ is at most $k$.
It follows that the order of the block derivation $D$ is at most $k$.

Finally, let $A \in \calB(D)$. We claim that $\bw(E(A)) \leq k$, which implies that the width of $D$ is at most~$k$.
Suppose without loss of generality
that $A$ is contained in $X_1$. We construct a branch-decomposition $\calT_A$ 
of $E(A)$ as follows. Let $\calT_1$ denote the sub-branch-decomposition
of $\calT$ at $p_1$. Since $\calT$ is injective, the width of $\calT_1$ is at most the width of
$\calT$, by \cref{prop:sub-bd}, and hence at most $k$. We obtain $\calT_A$ 
by removing leaves of $\calT$ mapped to edges not in $E(A)$ one by one. 
Since each edge removed intersects $A$ only in $M_1 = \bound(L_1)$,
\cref{prop:bd-edge-deletion} applies and the width of $\calT_A$ is at most the
width of $\calT_1$, which is at most $k$. Therefore, $\bw(E(A)) \leq k$
for every $A \in \calB(D)$.
\end{proof}

\subsection{Algorithm}
Algorithm~\ref{alg:bw-graph} shows our algorithm for deciding whether $\bw(G) \leq k$ for
a given graph $G$ and a positive integer $k$. 

\begin{algorithm}[tb]
    \KwIn{A graph $G$ and an integer $k \ge 0$.}
    \caption{Decide if $\bw(G) \le k$.}\label{alg:bw-graph}
        Compute the set $\calB_k(G)$ of all blocks $B$ in $G$ such that $|\bound(B)| \leq k$\;
        \For{$1 \le i \le n$}{
            $\calB_k(G, i) \leftarrow \{B \in \calB_k(G) : |\core(B)| = i\}$\;
        }
        \For{$B \in \calB_k(G)$}{
            Compute the set $\calD_k(G, B)$ of all block derivations of $B$ of order at most~$k$\;
        }
        \For{$1 \le i \le n$}{
            $\calF_k(G, i) \leftarrow \emptyset$\;
            \For{$B \in \calB_k(G, i)$}{
                \If{there is some $(B, \mathcal B) \in \calD_k(G, B)$ such that $A \in \bigcup_{j < i}\calF_k(G, j)$ for all $A \in \calB$}{
                    $\calF_k(G, i) \leftarrow \calF_k(G, i) \cup \{B\}$\;
                }
            }
        }
        Answer YES if $V(G) \in \calF_k(G, n)$; answer NO otherwise\;
\end{algorithm}



    

\begin{lemma}
    \label{prop:alg-correct}
    Algorithm \ref{alg:bw-graph} correctly decides if $\bw(G) \leq k$.
\end{lemma}
\begin{proof}
    We prove by induction on $i$ that $\calF_k(G, i)$ contains all blocks $B$ of $G$
    such that $\bw(E(B)) \leq k$ and $|\core(B)| = i$, for $1 \le i \le n$.  
    Fix $i \geq 1$ and suppose, for $1 \le j \le i - 1$,
    $\calF_k(G, j)$ contains all blocks $B$ in $G$
    such that $\bw(E(B)) \leq k$ and $|\core(B)| = j$.
    By \cref{thm:bd-iff}, $\bw(E(B)) \leq k$ if and only if there is $(B, \calB) \in
    \calD_k(G, B)$ such that $\bw(A) \leq k$ for
    every $A \in \calB$. By the induction hypothesis, 
    $\bw(E(A)) \le k$ holds if and only if $A \in \bigcup_{j < i} \calF_k(G, j)$.
    This completes the induction proof and $\calF_k(G, n)$ is the set of blocks $B$
    of $G$ such that $\bw(B) \leq k$. Therefore, the last step of the algorithm
    correctly decides if $\bw(G) \leq k$.
\end{proof}

\smallskip

Our running time analysis of the algorithm mostly consists in bounding the cardinality of 
$\calD_k(G) \coloneqq \bigcup_{B \in \calB_k(G)} \calD_k(G, B)$ and the time for generating those block derivations.
The following straightforward bound turns out to be useful in the analysis.

\begin{lemma}
    \label{prop:bd-naive-bound}
    There are at most $3^n$ block derivations of $G$ and they can be generated in time $\bigOh^*(3^n)$.
\end{lemma}
\begin{proof}
    For disjoint subsets $K$ and $S$ of $V(G)$, observe that there is at most one block derivation $D$
    such that $K(D) = K$ and $S(D) = S$. 
    Indeed, for such a block derivation $D = (B, \calB)$, it must hold that
    $B = K \cup S$ and $\calB = \{\hat{C} \mid C \in \calC\}$, where
    $\calC$ is the set of components of $G[K]$.
    It can be decided in polynomial time
    if the pair $(K, S)$ gives rise to such a block derivation, by checking if $N(C) \subseteq S$ for all 
   $C \in \calC$.
\end{proof}

We need an efficient way to decide if the width of a given block derivation is at most~$k$.
In particular, we need to compute a mid-triple $(M_1, M_2, M_3)$ of a block derivation $D$ of order at most~$k$.

\begin{observation}
    \label{prop:bd-order}
    If the order of a block derivation $D$ is $k$ then
    $|S(D)|$ is at most $3k/2$.
\end{observation}
\begin{proof}
    Let $(M_1, M_2, M_2)$ be a mid-triple of $D$ of order $k$.
    Then, since every vertex in $S(D)$ belongs to at least two of $M_1$, $M_2$, and $M_3$,
    we have $|S| \leq (|M_1| + |M_2| + |M_3|) / 2 \leq 3k / 2$.
\end{proof}

Due to this observation, we can focus on the case where $|S(D)| \leq 3k / 2$
in deciding if the width of $D$ is at most $k$.
We use the following tool to reason about mid-triples. Let $S$ be a set. A triple $(S_1, S_2, S_3)$
of subsets of $S$ is a \emph{bicover} of $S$ if each member of $S$ appears in exactly two of
$S_1$, $S_2$, and $S_3$. 
For a bicover $(S_1, S_2, S_3)$ of $S$, it holds that $S_1 \cup S_2 = S_2 \cup S_3 = S_3 \cup S_1 = S$.
It is a \emph{subbicover} of $S$ if each member of $S$ appears in at most two
of $S_1$, $S_2$, and $S_3$. We say that a bicover $(S_1, S_2, S_3)$ \emph{extends} a subbicover $(S_1', S_2', S_3')$
if $S_i' \subseteq S_i$ for $1 \le i \le 3$. The \emph{order} of a bicover $(S_1, S_2, S_3)$ is the largest of
$|S_1|$, $|S_2|$, and $|S_3|$. Since $|S_1| + |S_2| + |S_3| = 2 |S|$, the order of a bicover of $S$ is 
at least $2 |S| / 3$.

\begin{lemma}
    \label{lem:bicover}
    Let $S$ be a set and $k$ a positive integer such that $|S| \leq 3k / 2$.
    Let $(S_1, S_2, S_3)$ be a subbicover of $S$. Then, there is a bicover of $S$ of order at most $k$
    that extends $(S_1, S_2, S_3)$ if and only if the following conditions are satisfied.
    \begin{enumerate}
        \item $|S_i| \leq k$ for $i = 1, 2, 3$,
        \item $|S_i \cap S_j| + |S| \leq 2k$ for $1 \leq i < j \leq 3$;
    \end{enumerate}
\end{lemma}
\begin{proof}
    Suppose first that there is bicover $(R_1, R_2, R_3)$ of order $k$ that extends $(S_1, S_2, S_3)$.
    Then, for $1 \le i \le 3$, we have $|S_i| \leq |R_i| \leq k$. Since $|R_1| + |R_2| =
    |R_1 \cap R_2| + |R_1 \cup R_2| = |R_1 \cap R_2| + |S| $, we have $|S_1 \cap S_2| + |S| \leq |R_1 \cap R_2| + |S| \leq 2k$. Similarly, we have $|S_2 \cap S_3| + |S| \leq 2k$ and
    $|S_3 \cap S_1| + |S| \leq 2k$.

    For the converse, let $(S_1, S_2, S_3)$ be a subbicover of $S$ that satisfies the conditions
    in the lemma. First suppose that $S_1 \cup S_2 \cup S_3 = S$.
    The case otherwise will be dealt with later.
    For $1 \le i \le 3$, let $U_i$ be the set of members of $S$
    that belongs to $S_i$ but not to $S_j$ for $j \neq i$.
    We prove that there is bicover of $S$ of order at most $k$
    extending $(S_1, S_2, S_3)$ by induction on
    $|U_1 \cup U_2 \cup U_3|$. The base case where $U_1 \cup U_2 \cup U_3$ is empty is trivial since $(S_1, S_2, S_3)$ is then a bicover and, moreover, of order at most $k$ by Condition~1.
    Suppose that $U_1 \cup U_2 \cup U_3$ is non-empty.
    First suppose that Condition~2 is not tight for any pair of distinct $i, j \in \{1, 2, 3\}$:
    we have $|S_i \cap S_j| + |S| < 2k$ for every such pair.
    Without loss of generality, suppose $U_1$ is non-empty and 
    let $u$ be an arbitrary member of $U_1$.
    Then, we have $|S_2| + |S_3| = |S_2 \cup S_3| + |S_2 \cap S_3| =
    |S \setminus U_1| + |S_2 \cap S_3|
    \leq |S| + |S_2 \cap S_3| - 1
    \leq 2k - 1$ by Condition~2. 
    Note that each member of $S \setminus U_1$ is either not in $S_1$ or in $S_1 \cap S_i$ for some $i \in \{2,3\}$, and vise versa, meaning that $S_2 \cup S_3 = U \setminus S_1$ under the assumption that $S_1 \cup S_2 \cup S_3 = S$.
    Therefore, we have either 
    $|S_2| \leq k - 1$ or $|S_3| \leq k - 1$.
    We set $S_2' = S_2 \cup \{u\}$ and $S_3' = S_3$ in the former case and $S_2' = S_2'$ and $S_3' = S_3 \cup \{u\}$ in the latter. Then, Condition~1 with $S_i$ replaced by $S_i'$ for
    $i \in \{2, 3\}$ is maintained. Moreover, due to the assumption that Condition~2 is not tight for
    any pair of $i$ and $j$, Condition~2 with $S_i$ replaced by $S_i'$ for
    $i \in \{2, 3\}$ is also maintained.  Therefore, by the induction hypothesis, 
    there is a bicover of $S$ of order at most $k$ that extends $(S_1, S_2', S_3')$ and
    hence extends $(S_1, S_2, S_3)$. 

    Now suppose that $|S_1 \cap S_2| + |S| = 2k$. 
    Then, $|S_1 \cap S_2| \geq k / 2$
    since $|S| \leq 3k / 2$.
    Let $R_3 = S \setminus (S_1 \cap S_2)$. Note that $S_3 \subseteq R_3$ since $(S_1, S_2, S_3)$ is a subbicover and hence $S_1 \cap S_2 \cap S_3$ is empty.
    Note also that $|R_3| = |S| - |S_1 \cap S_2| \le 3k/2 - k/2 \leq k$ since $|S_1 \cap S_2| \ge k/2$ and $|S| \leq 3k / 2$.
    We define $R_1$ and $R_2$ so that 
    $(R_1, R_2, R_3)$ is a bicover of $S$  that extends $(S_1, S_2, S_3)$.
    To this end, we need to put each member of $U_3' = R_3 \setminus (S_1 \cup S_2)$
    to exactly one of $R_1$ and $R_2$.
    Regardless of our choices, we have $|R_1| + |R_2| = |R_1 \cup R_2| + |R_1 \cap R_2| = |S| + |S_1 \cap S_2| = 2k$. Since $|S_1| \leq k$ and $|S_2| \leq k$, we can partition $U_3'$ into
    $R_1 \setminus S_1$ and $R_2 \setminus S_2$ so that $|R_1| = |R_2| = k$.
    The triple $(R_1, R_2, R_3)$ thus obtained is a bicover of $S$ of order $k$ that extends $(S_1, S_2, S_3)$. The cases where $|S| + |S_2 \cap S_3| = 2k$ or
    $|S| + |S_3 \cap S_1| = 2k$ are similar.
    This concludes the induction step and the proof of the converse under the assumption that $S_1 \cup S_2 \cup S_3 = S$. 

    We finally deal with the case where $Q = S \setminus (S_1 \cup S_2 \cup S_3)$ is non-empty. Since $(S_1, S_2, S_3)$ is a subbicover of $S \setminus Q$, we have 
    $|S_1| + |S_2| + |S_3| \leq 2|S \setminus Q|$ and hence $|S_1| + |S_2| + |S_3| + |Q| 
    \leq |S_1| + |S_2| + |S_3| + 2|Q| \leq 2|S|
    \leq 3k$. Therefore, $Q$ can be partitioned into $Q_1$, $Q_2$, and $Q_3$ so that
    $|S_i \cup Q_i| \leq k$ for $i = 1, 2, 3$.
    Let $S_i' = S_i \cup Q_i$ for $i = 1, 2, 3$.
    Then, $(S_1', S_2', S_3')$ is a subbicover
    of $S$ satisfying the condition in the lemma
    with $S_i$ replaced by $S_i'$ for $1 \le i \le 3$. Note here that $S_1' \cap S_2' = S_1 \cap S_2$ and so on. Therefore, there is
    a bicover of $S$ of order at most $k$ that
    extends $(S_1', S_2', S_3')$ and hence extends $(S_1, S_2, S_3)$.   
\end{proof}
Let $D = (B, \calB)$ be a block derivation and let $(\calB_1, \calB_2)$ be a bipartition of $\calB$.
Let $\bound(\calB_i)$ denote $\bigcup_{A \in \calB_i} \bound(A)$
for $i = 1, 2$. We say that a mid-triple $(M_1, M_2, M_3)$ \emph{conforms to} the bipartition
$(\calB_1, \calB_2)$ if $\bound(\calB_i) \subseteq M_i$ for $i = 1, 2$.

\begin{lemma}
    \label{lem:bipartition2order}
    Let $D = (B, \calB)$ be a block derivation of a block $B$ in a graph $G$ and
     let $(\calB_1, \calB_2)$ be a bipartition of $\calB$. 
     Then, it can be decided in polynomial time if there is a mid-triple of $D$ of order at most $k$
     that conforms to $(\calB_1, \calB_2)$.
\end{lemma}
\begin{proof}
Let $(M_1, M_2, M_2)$ be a mid-triple of $D$ conforming to $(\calB_1, \calB_2)$.
We have $\bound(\calB_i) \subseteq M_i$ for $i = 1, 2$ and, moreover,
$\bound(B) \subseteq M_3$ as it is a condition in the definition of a mid-triple of $D$. 
Let $\delta_i = \bound(\calB_i)$ for $i = 1, 2$ and $\delta_3 = \bound(B)$. Let $I
= \delta_1 \cap \delta_2 \cap \delta_3$.
We may assume that
$M_1 \cap M_2 \cap M_3 = I$, since if some $v \in M_1 \cap M_2 \cap M_3$ does not belong to $I$, then we may remove $v$ from $M_i$ where $i$ 
is such that $v \not\in \delta_i$, maintaining the property of $(M_1, M_2, M_3)$ being a mid-triple of $D$. 
Let $S' = S \setminus I$,
$R_i = M_i \setminus I$ for $1 \le i \le 3$, and
$S_i = \delta_i \setminus I$ for $1 \le i\le 3$.
Then, $(R_1, R_2, R_3)$ is
a bicover of $S'$ that extends the subbicover $(S_1, S_2, S_3)$ of $S'$.
As the order of $(R_1, R_2, R_3)$ is the order of
$(M_1, M_2, M_3)$ minus $|I|$, the problem of deciding if $D$ has a mid-triple of order $k$ reduces to deciding if there is a bicover of $S'$
of order at most $k - |I|$
that extends $(S_1, S_2, S_3)$, which can be immediately answered by checking the conditions in \cref{lem:bicover}.
\end{proof}

By~\cref{lem:bipartition2order}, when $\calB(D)$ contains a small number of sets, we can efficiently compute a mid-triple $(M_1, M_2, M_3)$ of $D$ of order at most $k$ by trying all possible bipartitions $(\calB_1, \calB_2)$ of $\calB(D)$ to which $(M_1, M_2, M_3)$ conforms.
For the case where $\calB(D)$ contains many sets, we have an alternative method of deciding 
if the width of $D$ is at most $k$. We call a mid-triple $(M_1, M_2, M_3)$ of a block derivation $D = (B, \calB)$ 
\emph{$k$-full} if $|M_1| = |M_2| = k$.
\begin{lemma}
    \label{prop:block-deriv-intersec}
Let $D = (B, \calB)$ be a block derivation of a block $B$ of order at most $k$ 
such that $|S(D)| \geq k$.
Then, there is a $k$-full mid-triple of $D$ of order $k$.
\end{lemma}
\begin{proof}
Let $(M_1, M_2, M_3)$ be a mid-triple of $D$ of order at most $k$.
Initialize $M_i'$ to $M_i$ for $1 \le i \le 3$.
While $|M_1'| < k$, add a vertex in $M_2' \setminus M_1'$ to $M_1'$.
Similarly, while $|M_2'| < k$, add a vertex in $M_1' \setminus M_2'$ to $M_2'$.
In this process, $(M_1', M_2', M_3')$ remains a mid-triple of $D$
of order at most $k$. Since $|S(D)| = |M_1 \cup M_2| \geq k$, 
the process stops with $|M_1'| = |M_2'| = k$.
\end{proof}

\begin{lemma}
    \label{prop:k-full}
    Let $(M_1, M_2, M_3)$ be a $k$-full mid-triple of a block derivation
$D = (B, \calB)$. Then $|M_1 \cap M_2| = 2k - |S(D)|$.
\end{lemma}
\begin{proof}
   As $S(D) = M_1 \cup M_2$ by the definition of a mid-triple,
   we have $|S(D)| + |M_1 \cap M_2| = |M_1| + |M_2| = 2k$.
\end{proof}

\begin{lemma}
    \label{lem:intersection2order}
    Let $D = (B, \calB)$ be a block derivation in $G$, $I$ a subset of $S(D)$, 
    and $k$ a positive integer such that $|I| = 2k - |S(D)|$.
    Then, it can be decided in polynomial time if there is
    a $k$-full mid-triple $(M_1, M_2, M_3)$ of $D$ of order $k$ such that $M_1 \cap M_2 = I$.
\end{lemma}
\begin{proof}
Let $A \in \calB$. Observe that, in a mid-triple $(M_1, M_2, M_3)$ of 
$D$ such that $M_1 \cap M_2 = I$, the vertices in $\bound(A) \setminus I$
either altogether belong to $M_1$ or altogether belong to $M_2$.
This observation motivates a graph $Q$ on vertex set $S(D) \setminus I$ such that
there is an edge between $u$ and $v$ if and only if $u$ and $v$ 
belong to $\bound(A)$ for some $A \in \calB$. 
Then, each connected component of $Q$
must be contained in exactly one of $M_1$ or $M_2$. Let $\gamma_1$, \ldots, $\gamma_m$
be the list of connected components of $Q$. 
Then, our problem of finding 
a $k$-full mid-triple $(M_1, M_2, M_3)$ of order $k$ such that $M_1 \cap M_2 = I$
reduces to a subset sum problem that asks for a subset $J$ of $\{1, \ldots, m\}$
such that $\sum_{j \in J} |\gamma_j| =  k - |I|$.
By the classical Bellman's dynamic programming algorithm, this problem
can be solved in time $\bigOh(m (k - |I|)) \subseteq \bigOh(n^2)$. 
Setting up the input for the subset sum problem can be done in polynomial time, and thus the lemma follows.
\end{proof}

\begin{lemma}
\label{lem:bd-order}
    Let $D = (B, \calB)$ be a block derivation in a graph $G$ with $|S(D)| \geq k$. 
    Then, it can be decided in time 
    $\bigOh^*\left(\min\left\{2^{|B| - |S(D)|}, \binom{|S(D)|} {2k - |S(D)|} \right\}\right)$
    whether the order of $D$ is at most $k$.
\end{lemma}
\begin{proof}
The first bound of $\bigOh^*(2^{|B| - |S(D)|})$ is achieved by a brute force method 
that tries all bipartitions
$(\calB_1, \calB_2)$ of $\calB$ deciding if there is a mid-triple $(M_1, M_2, M_3)$
of $D$ of order at most $k$ that conforms to this bipartition in polynomial time invoking 
\cref{lem:bipartition2order}.
For the second bound, by~\cref{{prop:block-deriv-intersec}}, it suffices to find a $k$-full mid-triple of $D$ of order at most~$k$.
We try all $\binom{|S|} {2k - |S|}$ subsets $I$ of $S$ of cardinality $2k - |S|$
deciding, for each such subset $I$, if there is a $k$-full mid-triple $(M_1, M_2, M_3)$ of order $k$ such that $M_1 \cap M_2 = I$ in polynomial time invoking \cref{lem:intersection2order}.
\end{proof}

We now upper bound the number of blocks and the number of block derivations
needed to decide if $\bw(G)$ is at most $k$.
We use the following result due to Fomin and Villanger~\cite{FominV12}.

\begin{lemma}[Fomin and Villanger~\cite{FominV12}]
\label{lem:Fomin-Villanger}
    Let $G$ a graph and $b$, $c$ positive integers. Then, 
the number of connected vertex subsets $C$ of $G$ such that
$|C|= c$ and $|N_G(C)| = b$ is at most $n\binom{b + c - 1}{b}$ and those
subsets can be listed in time $\bigOh^*\left(\binom{b + c - 1}{b}\right)$.
\end{lemma}

\begin{lemma}
    \label{lem:block-count}
    Let $G$ be a graph with $n$ vertices and $a, b$ positive integers such that
    $b \leq a \leq n$.
    Then, the number of blocks $B$ such that 
    $|B| = a$ and $|\bound(B)| = b$ is at most $\min\left\{n\binom{a - 1}{b}, \binom{n}{a}\right\}$.
    Moreover, those blocks can be listed in time $\bigOh^*\left(\min\left\{\binom{a - 1}{b}, \binom{n}{a}\right\}\right)$.
\end{lemma}
\begin{proof}
    Since $\core(B)$ is connected and $\bound(B) = N(\core(B))$, 
    the number of blocks $B$ such that 
    $|B| = a$ and $|\bound(B)| = b$ is at most $n\binom{a - 1}{b}$ by \cref{lem:Fomin-Villanger}.
    The latter bound is simply the number of subsets of $V(G)$ of cardinality $a$.
    The listing procedure in \cref{lem:Fomin-Villanger} together with
    a brute force procedure for the latter bound achieves the claimed time bound.
\end{proof}

Recall that $\calD_k(G) = \bigcup_{B \in \calB_k(G)} \calD_k(G, B)$, where $\calB_k$ is the set of all blocks $B$ with $|\bound(B)| \le k$.
\begin{lemma}
\label{lem:bd-bound}
    Let $f(n, k, s, a, b)$ be the function defined as follows.
    \begin{align*}
    f(n, k, s, a, b)  &\coloneqq  \min\left\{\binom{a - 1}{b}, \binom{n}{a}\right\} 
    \cdot \binom{a - b}{s - b} \cdot \min\left\{2^{a - s}, \binom{s} {2k - s}\right\},
    \end{align*}
    where the domain of each variable consists of integers that satisfy the following: $0 \le k \le n$, $k < s \le 3k/2$, $s < a \le n$, and $b \le \min\{k, a-1\}$.
    Let $f^*(n, k)$ denote the maximum of $f(n, k, s, a, b)$ over all combinations of
    $s$, $a$, and $b$ in the domain.
    Then, $|\calD_k(G)| = \bigOh^*(f^*(n, k) + 3^n)$ and $\calD_k(G)$ can be 
    computed in $\bigOh^*(f^*(n, k) + 3^n)$ time.
\end{lemma}
\begin{proof}
    Fix $k$. Based on \cref{prop:bd-naive-bound}, We first list all block derivations in $G$ in 
    $\bigOh^*(3^n)$ time. For each block derivation $D$ listed, if $|S(D)| \leq k$ then
    put $D$ to $\calD_k(G)$, since such $D$ trivially has a mid-triple $(S(D), \emptyset, S(D))$
    of order $|S(D)|$. 
    
    We then pick up block derivations $D$ with $k < |S(D)| \leq 3k/ 2$ 
    that are of order at most~$k$ as follows.
    For each $b \leq k$ and $a > b$, we compute the set $\calB_{a, b}(G)$ of blocks $B$ such that 
    $|B| = a$ and $|\bound(B)| = b$, in time $\bigOh^*(\min\{\binom{a - 1}{b}, \binom{n}{n - a}\})$ due to \cref{lem:block-count}.
    For all the combinations of $s$, $a$, and $b$ such that $k < s \leq 3k/ 2$,
    $b \leq k$, $b < a$, and $s < a \leq n$, and for each $B \in \calB_{a, b}(G)$, 
    we try each subset $S$ of $B$ such that $|S| = s$ and $\bound(B) \subseteq S$, 
    letting $\calC$ be the set of components of $G[B \setminus S]$,
    letting $\calB = \{\hat{C} \mid C \in \calC\}$, and 
    deciding if the order of the block derivation $D = (B, \calB)$ with $S(D) = S$ is at most $k$
    using \cref{lem:bd-order}; we add $(B, \calB)$ to $\calD_k(G)$ if this is the case.
    We have $\binom{a - b}{s - b}$ choices for $S$ and deciding if $(B, \calB)$ is of order at most $k$ can be done in time $\bigOh^*(\min\{2^{a - s}, \binom{s} {2k - s}\})$. 
    The time for each combination of $a$, $b$, and $s$ is at most $f(n, k, s, a, b)$
    and hence the total time for processing all combinations is in $\bigOh^*(f^*(n, k))$. 
\end{proof}

\begin{theorem}
    \label{thm:bw-graph}
    The running time of Algorithm~\ref{alg:bw-graph} is in $\bigOh(3.293^n)$.
\end{theorem}
\begin{proof}
We can compute the set $\calB_k$ of all blocks $B$ in $G$ such that $|\bound(B)| \le k$ in $\bigOh^*(2^n)$ time. 
By \cref{lem:bd-bound}, $|\calD_k(G)| = \bigOh^*(f^*(n, k) + 3^n)$ and so is the running time of generating $\calD_k(G)$.
The dynamic programming step for computing $\calF_k(G, i)$ can also be done in time $\bigOh^*(f^*(n, k) + 3^n)$.

To evaluate $f^*(n, k)$,
we use the well known upper bound on the binomial coefficient that $\binom{n}{k} \leq 2^{h(k / n)}$, where $h(x)$ is the binary entropy function (see~\cite{FominK10}, for example).
The logarithm base two of the function $f$ in \cref{lem:bd-bound} is bounded as
\begin{align*}
    \log_2 f(n, k, s, a, b) &\leq n \cdot g(k/n, s/n, a/n, b/n),
\end{align*}
where 
\begin{align*}
    g(\kappa, \sigma, \alpha, \beta) = \min\{h(\beta / \alpha), h(\alpha)\} + h((\sigma - \beta) / (\alpha - \beta)) + \min\{\alpha - \sigma, h((2\kappa - \sigma) / \sigma)\}.
\end{align*}
A numerical computation shows that the maximum of $g(\kappa, \sigma, \alpha, \beta)$ in the domain corresponding to the domain for $f$ described in \cref{lem:bd-bound}
is at most 1.719253. The code for this computation, implemented by Claude Code (Opus 4.6), uses the interval arithmetic
provided by the npmath module of Python and therefore this upper bound is rigorous. As $2^{1.71953} < 3.2927$,
we have $f^*(n, k) \leq 3.2927^n$ and therefore the running time of the algorithm is 
in $\bigOh(3.293^n)$.
\end{proof}

This completes the proof of \cref{thm:main}.

\subsection{An alternative algorithm for graphs}
In this subsection, we present an alternative to Algorithm~\ref{alg:bw-graph}.
We conjecture that its running time is $\bigOh(c^n)$ where $c$
is a constant smaller than 3.293, the constant in the running time bound of Algorithm~\ref{alg:bw-graph}. See 
the partial analysis at the end of this subsection for some grounds of this conjecture.
Regardless of the validity of this conjecture, its practical performance is similar to Algorithm~\ref{alg:bw-graph} on 
most instances but significantly better on some instances that are hard for Algorithm~\ref{alg:bw-graph}. 

The idea of the algorithm is to restrict blocks $B$ for which we need to compute the branchwidth of $E(B)$.
Call an edge set $F$ of $G$ \emph{small} if $|\core(F)| \leq |V(G) \setminus V(F)|$.
Call a block $B$ \emph{small} if $E(B)$ is small.
Call a tripartition $(E_1, E_2, E_3)$ of $E(G)$ \emph{balanced} if
$E_i$ is small for $1 \le i \le 3$. 

\begin{proposition}
    \label{prop:balanced}
    Let $G$ be a graph with at least three edges and of minimum degree at least~2.
    Then, there is a balanced tripartition $(E_1, E_2, E_3)$ of $E(G)$
    such that $\bw(E_i) \leq \bw(G)$. 
\end{proposition}
\begin{proof}
    Let $\calT = (T, \varphi)$ be an injective branch-decomposition of $E(G)$ of width $\bw(G)$.
    Let $T'$ be an undirected ternary tree\footnote{A \emph{ternary tree} is a (unrooted) tree such that every internal node has degree exactly~3.} obtained from $T$ by replacing the
    two edges incident to the root by an edge. Let $p$ be an internal node of $T'$.
    The removal of $p$ from $T'$ naturally induces a tripartition of $E(G)$, which
    we call the tripartition associated with $p$. We show that there is an internal node $r$ of $T'$
    such that the tripartition associated with $r$ is balanced. 
    For each edge $\{p, q\}$ of $T'$, orient this edge from
    $p$ to $q$ if the edge set on $p$'s side is small; from $q$ to $p$ if the edge set
    on $q$'s side is small, breaking ties arbitrarily.
    Under our assumption on $G$, each edge incident to a leaf of $T'$ is oriented away from the leaf.
    Since every edge is oriented, there is a node $r$ such that all three incident edges
    are oriented towards $r$. The tripartition associated with $r$ is balanced.
    The subtrees of $T'$ that result from removing $r$ are
    branch-decompositions of $E_1$, $E_2$, and $E_3$. It is straightforward to see
    that the width of each of them is at most the width of $\calT$.
\end{proof}

In the following, we use similar but slightly different definitions from those used in the previous subsection.
A \emph{root derivation} of a graph $G$ is a set $\calB$ of blocks of $G$
such that $\core(A_1) \cap \core(A_2) = \emptyset$ for each pair of distinct $A_1, A_2 \in \calB$. We let $S(\calB) = V(G) \setminus \core(\calB)$, where $\core(\calB) = \bigcup_{A \in \calB} \core(A)$ as before.
A \emph{mid-triple} of a root derivation $\calB$ is a triple $(M_1, M_2, M_3)$ of
subsets of $S(\calB)$ such that $M_1 \cup M_2 = M_2 \cup M_3 = M_3 \cup M_1 = S(\calB)$.
Let $(\calB_1, \calB_2, \calB_3)$ be a tripartition of $\calB$. We say
that a mid-triple $(M_1, M_2, M_3)$ \emph{conforms to} this tripartition
if $\bound(\calB_i) \subseteq M_i$ for $1 \le i \le 3$.
A mid-triple of $\calB$ is \emph{conforming} if it conforms to some
tripartition of $\calB$. 
The order of a mid-triple $(M_1, M_2, M_3)$ is the largest of
$|M_i|$ for $1 \le i \le 3$. The \emph{order} of a root derivation $\calB$ is
the smallest $k$ such that there is a conforming mid-triple of $\calB$ of order $k$.

\begin{lemma}
    \label{lem:root-deriv}
    Let $G$ be a graph with at least two edges.
    Then, $\bw(G) \leq k$ if and only if there is a root derivation
    $\calB$ of order at most $k$ such that, for each $A \in \calB$, 
    $A$ is small and $\bw(E(A)) \leq k$.
\end{lemma}
\begin{proof}
    The proof of the ``if''-part is similar to the proof of \cref{lem:block-deriv2bw} on block derivations. For the converse, suppose $\bw(G) = k$. By \cref{prop:balanced},
    there is a balanced tripartition $(E_1, E_2, E_3)$ of $E(G)$ such that
    $\bw(E_i) \leq k$ for $1 \le i \le 3$. Let $M_i = \bound(E_i)$ for $1 \le i \le 3$ and
    let $S = M_1 \cup M_2 \cup M_3$. Note that $S = M_1 \cup M_2 = M_2 \cup M_3 = M_3 \cup M_1$.
    Let $\calC$ denote the set of components of $G[V(G) \setminus S]$ and
    let $\calB = \{\hat{C} \mid C \in \calC\}$. Observe that $\bound(A)$ for each $A \in \calB$
    is contained in $M_i$ for some $i \in \{1, 2, 3\}$.
    Let $(\calB_1, \calB_2, \calB_3)$ be
    the tripartition of $\calB$ obtained by putting each $A \in \calB$ to
    $\calB_i$ where $i$ is the smallest such that $\bound(A) \subseteq M_i$. 
    Then, $(M_1, M_2, M_3)$ is a mid-triple of $\calB$ conforming to the tripartition
    $(\calB_1, \calB_2, \calB_3)$. As the order of this mid-triple is at most $k$,
    the order of the root derivation $\calB$ is at most $k$. 
    Each block in $\calB_i$ is small since $E_i$ is small. Finally, for each $A \in \calB_i$,
    $\bw(A) \leq k$ holds, since a branch-decomposition of $E(A)$ of width at most $k$ can
    be obtained from the branch-decomposition of $E_i$ of width at most $k$ by removing 
    leaves that map to edges not in $E(A)$.
\end{proof}

Algorithm~\ref{alg:bw-alternative} shows a pseudo-code of our alternative algorithm.

\begin{algorithm}[tb]
    \KwIn{A graph $G$ and an integer $k \ge 0$.}\caption{Decides if $\bw(G) \leq k$.} \label{alg:bw-alternative}
    Compute the set $\calF_k(G)$ of all small blocks $A$ of $G$ such that $\bw(E
(A)) \leq k$, using the method of Algorithm~\ref{alg:bw-graph}\;
    Generate root derivations $\calB$ of $G$ such that $\calB \subseteq \calF_k(G)$
    and $|S(\calB)| \leq 3k/ 2$\;
    \For{$\calB$ generated}{
        \If{the order of $\calB$ is at most $k$}{
            Answer YES and halts\;
        }
    }
    Answer NO\;
\end{algorithm}
The correctness of this algorithm follows straightforwardly from \cref{lem:root-deriv}.
We conjecture that this algorithm can be implemented in such a way to run in time 
smaller than the upper bound given in \cref{thm:bw-graph} on Algorithm~\ref{alg:bw-graph}. 
The following is a partial analysis of the algorithm that provides some grounds for the conjecture.

\begin{enumerate}
\item
It is intuitively clear that the running time of generating $\calF_k(G)$ is much smaller than that of 
Algorithm~\ref{alg:bw-graph}, since we compute the branchwidth of small blocks only. 
Indeed, it is not difficult to show that the running time of this step is in $\bigOh^*(2.88^n)$.
\item
For each $S \subseteq V(G)$, let $\calC(S)$ denote the set of components of
$G[V(G) \setminus S]$ and let $\calB(S) = \{\hat{C} \mid C \in \calC(S)\}$.
Let $\calS_k(G)$ denote the family of vertex sets $S$ of $G$ such that
$|S| \leq 3k / 2$ and, for every $A \in \calB(S)$, $A$ is small and $\bw(A) \leq k$.
Step 2 of the algorithm looks for $S \in \calS_k(G)$ such that the order of
the root derivation $\calB(S)$ is at most $k$. Let $Q(S)$ be a hypergraph on $S$ obtained from $G[S]$ by adding $\bound(A)$ for every $A \in \calB(S)$
as a hyperedge. Then, deciding if the order of $\calB(S)$ is at most $k$
is equivalent to deciding if $\bw(Q) \leq k$. We can use our algorithm given in \cref{thm:main-hg}
spending $\bigOh^*(4^{|S|})$ time. We can also solve this problem in $\bigOh^*(3^{|\calB(S)|})$ time, deciding, for
each tripartition of $\calB(S)$, if there is a mid-triple of order at most $k$ that conforms to the tripartition.
Overall, the running time of the remaining steps is upper bounded by
$\bigOh^*\left(\sum_{S \in \calS_k(G)} \min\{4^{|S|}, 3^{|\calB(S)|}\}\right)$. Although the only upper bound we have
on $|\calB(S)|$ is the trivial $n - |S|$, this bound corresponds to an extremely rare case where
every component in $\calC(S)$ is a singleton. It would not be surprising if there is some useful 
combinatorial bound on the above sum.
We also note that, as $k \geq 2|V(Q)|/ 3$, deciding if the branchwidth of $Q$ is at most $k$ might admit an
algorithm that is faster than algorithms for general hypergraphs. 
This possibility would be a fascinating research topic. Positive results could lead to a
proof of the conjecture as well as to a dramatically improved practical performance of Algorithm~\ref{alg:bw-alternative}.
\end{enumerate}

\section{Computational experiment}
In this section, we present the results of our computational experiment.
Although our computational experiment is somewhat preliminary and further rigorous comparative evaluations, such as running existing algorithms in the exact same computational environment, are necessary, we can nevertheless observe that the proposed algorithms are highly efficient in practice.

In our experiment, we evaluated the SAT encoding approach~\cite{LodhaOS19} and our algorithm using ``named graphs''~\cite{DellKTW17}\footnote{\url{https://github.com/PACE-challenge/Treewidth}}, some of which have also been used in previous work~\cite{LodhaOS19}, and the instances for the exact track of the treedepth computation~\cite{KowalikMNPSW20}\footnote{\url{https://github.com/lkowalik/Treedepth-PACE-2020-instances/}} in the PACE Challenge repository. 
We implemented our proposed algorithms in Java\footnote{The source code is available at \url{https://github.com/twalgor/bw}.} and used the authors' implementation\footnote{\url{https://www.ac.tuwien.ac.at/research/branchlis/}} of the SAT-encoding approach due to \cite{LodhaOS19} with the Glucose 4.0 SAT solver.
The experiment was conducted on a machine equipped with AMD Ryzen Threadripper 3990X 64-Core Processor and 64 GB of RAM, running Ubuntu 22.04.5 LTS.
We ran \cref{alg:bw-alternative} (Ours) and the SAT-encoding approach (SAT) for each of the above instances in 10 minutes. 

The implementations of our algorithms are mostly faithful, although, for some subtasks,
they put priority on practical efficiency over theoretical time complexity. The most notable difference is in the
algorithm to decide if the order of a given block derivation is at most $k$. Our implementation 
deviates from the procedure described in the proof of \cref{lem:bd-order} in that it does not
adopt the second method that goes through all subsets of the separator $S$ of cardinality $2k - |S|$.
It rather performs a backtrack search to generate candidate mid-triples. Although the running time can be
larger than $\bigOh\left(\binom{|S|}{2k - |S|}\right)$, it seems to perform well in practice with some
simple pruning rules. We note, however, that it may be a good idea to incorporate, in future implementations,
a mechanism to detect on the fly an unaffordable explosion of the backtrack search and switch to the
theoretical method in \cref{lem:bd-order}.

\cref{fig:plot} shows the experimental results on the named graph instances (\texttt{named}) and the treedepth instances (\texttt{td}).
\cref{alg:bw-graph} and \cref{alg:bw-alternative} solve 110 and 112 instances out of 150 instances in \texttt{named} and 155 and 165 instances out of 200 instances in \texttt{td}, while the SAT-encoding approach solves 49 and 43 instances in \texttt{named} and \texttt{td}, respectively. 
Instance-wise speed-ups are also overwhelming. The minimum, maximum, and geometric average speed-ups of \cref{alg:bw-graph} over the SAT approach are approximately $3.2$, $3388$, and $32.7$ times among the instances that are solved by all the algorithms. Those numbers for \cref{alg:bw-alternative} 
are approximately $3.1$, $3362$, and $35.4$ times.
See Tables~\ref{tab:results_named}--\ref{tab:last} in \cref{apd:tables} for the full experimental results.

\begin{figure}
    \centering
    \includegraphics[width=0.49\linewidth]{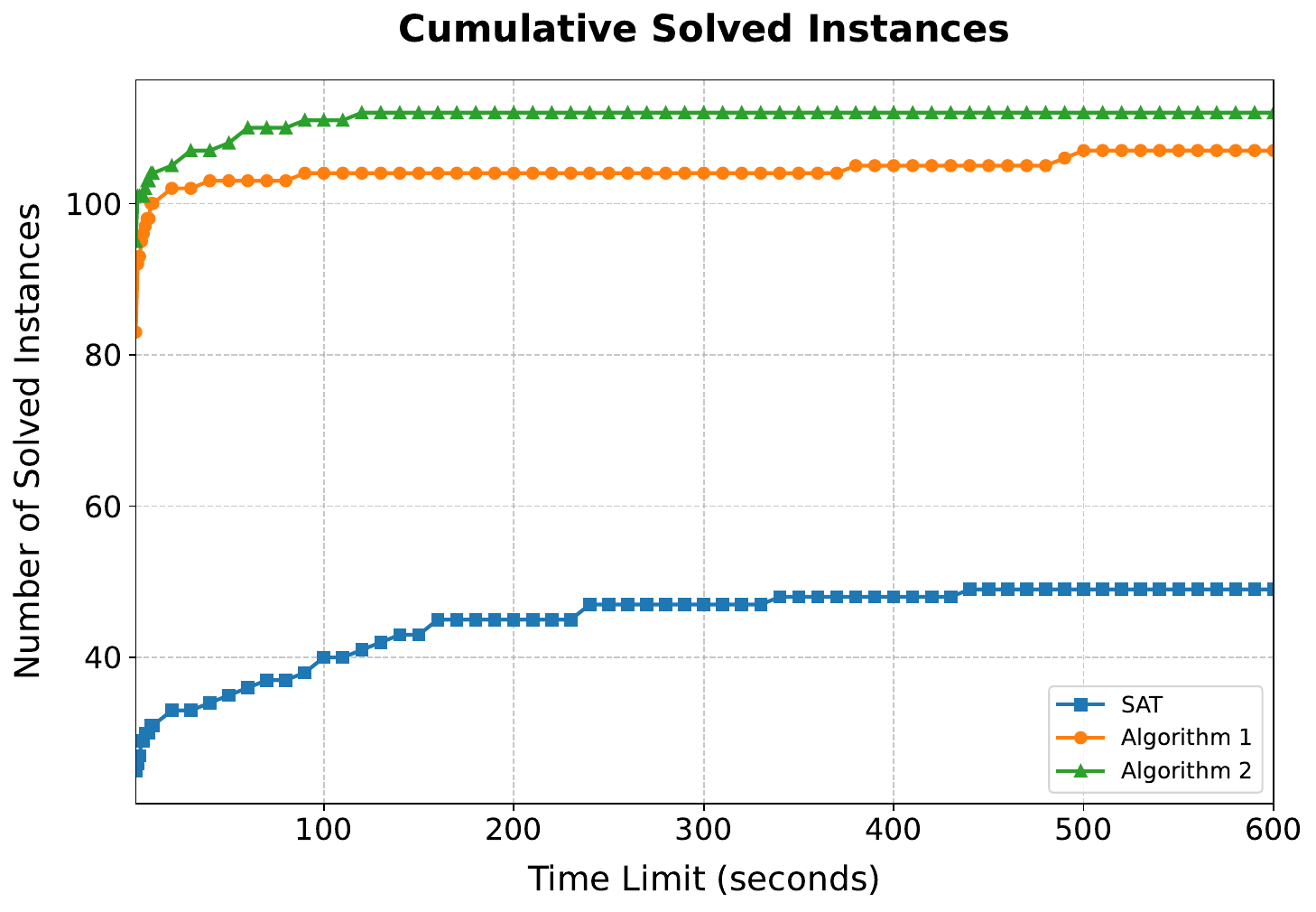}
    \includegraphics[width=0.49\linewidth]{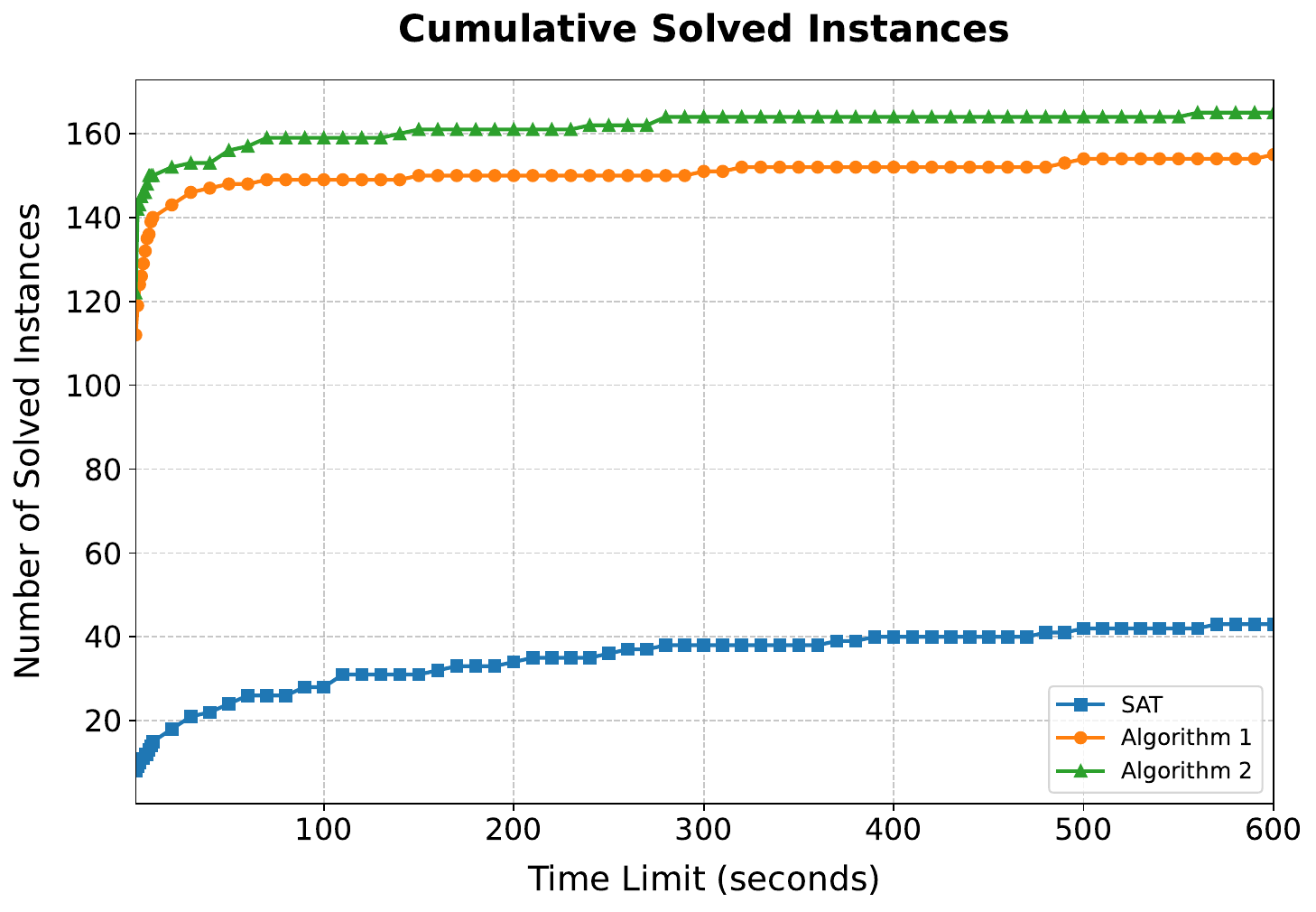}
    \caption{These plots illustrate how many instances each algorithm could solve within a given time limit. The left and right plots correspond to the results for named graphs and the treedepth instances, respectively.}
    \label{fig:plot}
\end{figure}

We also experimented on those instances that are used in \cite{LodhaOS19}. 
\cref{{tab:exp_results}} shows the results of our experiment alongside with the results reported in \cite{LodhaOS19}. As \cite{LodhaOS19} contains comparisons between the SAT-encoding approach and another state-of-the-art practical algorithm due to Hicks~\cite{Hicks05}, our table allows us to indirectly compare our algorithms with Hicks's. Despite the difference of
the computational platforms used, we can safely say that our algorithms overwhelmingly outperform Hicks's Algorithm as well, with the exception that, on small instances with at most 40 vertices and branchwidth at most 5,
Hicks's algorithm is mostly comparable to, and sometimes outperforms, ours.
\begin{table}[p]
\begin{center}
\caption{Experimental results for instances shown in Table~3 in \cite{LodhaOS19}. We excluded the instance ``Pmin'' from this table, as its detail is unclear.
The columns ``Alg~1'', ``Alg~2'', and ``SAT'' are the results for \cref{alg:bw-graph}, \cref{alg:bw-alternative}, and the SAT-encoding of \cite{LodhaOS19} executed on our computational environment and the columns ``SAT*'' and ``Tangle*'' are those for the SAT-encoding and the tangle-based algorithm of \cite{Hicks05} that are reported in \cite{LodhaOS19}.
Note that the results in the two rightmost columns are obtained in a different computational environment with time limit 2,000 seconds from the others.
The entries ``*'' and MO indicate that the branchwidth can be computed within 24 hours and cannot be computed due to an out-of-memory error, respectively.
}
\label{tab:exp_results}
\begin{tabular}{ c|c|c|c|c|c|c|c|c }
\hline 
instance & $|V|$ & $|E|$ & $\bw$ & Alg~1 & Alg~2 & SAT & SAT* & Tangle*\\ 
 \hline
Ellingham       & 78 & 117 & 6  & 2.364  & 1.079  & --      & -- & -- \\
B10Cage         & 70 & 106 & -- & --     & --     & -- & -- & -- \\
Watkin          & 50 & 75  & 6  & 1.731  & 0.624  & --      & -- & -- \\
Paley17         & 17 & 68  & 10 & 0.122  & 0.121  & --      & -- & -- \\
Kittell         & 23 & 63  & 6  & 0.123  & 0.123  & --      & -- & 315.88\\
Holt            & 27 & 54  & 9  & 0.844  & 0.259  & --      & * & MO\\
Shrikhande      & 16 & 48  & 8  & 0.120  & 0.121  & --      & * & * \\
Errera          & 17 & 45  & 6  & 0.108  & 0.118  & --      & 1,384.21 & 3.70\\
Brinkmann       & 21 & 42  & 8  & 0.248  & 0.152  & --     & -- & 3,163.38\\
5x5-grid        & 25 & 40  & 5  & 0.188  & 0.159  & 233.19 & 13.45 & 0.06\\
Folkman         & 20 & 40  & 6  & 0.127  & 0.128  & 430.36 & 1,518.40 & 3.40\\
Clebsch         & 16 & 40  & 8  & 0.126  & 0.123  & --      & -- & 2,017.62\\
Poussin         & 15 & 39  & 6  & 0.117  & 0.115  & --      & 645.63 & 4.75\\
Paley13         & 13 & 39  & 7  & 0.111  & 0.112  & --      & -- & 179.11\\
Robertson       & 19 & 38  & 8  & 0.168  & 0.144  & --      & -- & 3,906.94\\
McGee           & 24 & 36  & 7  & 0.929  & 0.282  & --      & -- & 177.92\\
Nauru           & 24 & 36  & 6  & 0.316  & 0.163  & 338.78 & 436.23 & 43.34\\
Hoffman         & 16 & 32  & 6  & 0.123  & 0.122  & 111.35 & 997.97 & 3.97\\
Desargues       & 20 & 30  & 6  & 0.204  & 0.159  & 123.13 & 474.72 & 2.10\\
Dodecahedron    & 20 & 30  & 6  & 0.261  & 0.159  & 239.67 & 152.96 & 2.29\\
Flower Snark    & 20 & 30  & 6  & 0.224  & 0.155  & 158.72 & 995.93 & 1.87\\
Pappus          & 18 & 27  & 6  & 0.210  & 0.147  & 86.39  & 241.12 & 1.76\\
Sousselier      & 16 & 27  & 5  & 0.124  & 0.118  & 5.528  & 18.47 & 0.12\\
Goldner--Harary & 11 & 27  & 4  & 0.102  & 0.106  & 2.444  & 9.97 & 0.48\\
4x4-grid        & 16 & 24  & 4  & 0.128  & 0.120  & 1.015  & 3.77 & 0.10\\
Chv\'atal       & 12 & 24  & 6  & 0.115  & 0.113  & 49.98  & 81.20 & 0.82\\
Gr\"otzsch      & 11 & 20  & 5  & 0.110  & 0.112  & 3.345  & 11.37 & 0.55\\
D\"urer         & 12 & 18  & 4  & 0.115  & 0.116  & 0.655  & 2.45 & 0.01\\
Franklin        & 12 & 18  & 4  & 0.114  & 0.113  & 0.598  & 4.85 & 0.04\\
Frucht          & 12 & 18  & 3  & 0.108  & 0.107  & 0.452  & 2.16 & 0.41\\
Tietze          & 12 & 18  & 4  & 0.113  & 0.115  & 0.662  & 1.52 & 0.04\\
Herschel        & 11 & 18  & 4  & 0.110  & 0.107  & 0.642  & 8.49 & 0.08\\
Petersen        & 10 & 15  & 4  & 0.111  & 0.113  & 0.513  & 1.62 & 0.09\\
3x3-grid        & 9  & 12  & 3  & 0.108  & 0.105  & 0.349  & 1.10 & 0.02\\
Wagner          & 8  & 12  & 4  & 0.104  & 0.110  & 0.401  & 1.54 & 0.02\\
Moser spindle   & 7  & 11  & 3  & 0.104  & 0.091  & 0.354  & 1.05 & 0.014\\
Prism           & 6  & 9   & 3  & 0.104  & 0.103  & 0.313  & 1.33 & 0.03\\
Butterfly       & 5  & 6   & 2  & 0.097  & 0.100  & 0.300  & 2.11 & 0.02\\
Bull            & 5  & 5   & 2  & 0.087  & 0.090  & 0.300  & 0.00 & 0.00\\
Diamond         & 4  & 5   & 2  & 0.097  & 0.089  & 0.301  & 2.09 & 52.54\\
\hline
\end{tabular}
\end{center}
\end{table}

\bibliography{references}

\newpage
\appendix

\section{Further experimental results}\label{apd:tables}
\begin{table}[H]
    \centering
    \caption{Performance comparison on named graphs (1/4).}
    \label{tab:results_named}
    \small
    \begin{tabular}{lrrrrrr}
        \toprule
        Instance & $|V|$ & $|E|$ & $\textrm{bw}$ & SAT & Alg 1 & Alg 2 \\ \midrule
        AhrensSzekeresGeneralizedQuadrangleGraph\_3 & 27 & 135 & 16 & -- & 0.277 & 0.383 \\
        Balaban10Cage & 70 & 105 & -- & -- & -- & -- \\
        Balaban11Cage & 112 & 168 & -- & -- & -- & -- \\
        BalancedTree\_3\_5 & 364 & 363 & 2 & -- & 0.139 & 0.135 \\
        BarbellGraph\_10\_5 & 25 & 96 & 7 & -- & 0.113 & 0.113 \\
        BidiakisCube & 12 & 18 & 4 & 0.624 & 0.115 & 0.113 \\
        BiggsSmithGraph & 102 & 153 & -- & -- & -- & -- \\
        BlanusaFirstSnarkGraph & 18 & 27 & 5 & 8.163 & 0.14 & 0.135 \\
        BlanusaSecondSnarkGraph & 18 & 27 & 4 & 1.542 & 0.126 & 0.118 \\
        BrinkmannGraph & 21 & 42 & 8 & -- & 0.248 & 0.152 \\
        BrouwerHaemersGraph & 81 & 810 & -- & -- & -- & -- \\
        BubbleSortGraph\_5 & 120 & 240 & -- & -- & -- & -- \\
        BuckyBall & 60 & 90 & -- & -- & -- & -- \\
        CameronGraph & 231 & 3465 & -- & -- & -- & -- \\
        Cell120 & 600 & 1200 & -- & -- & -- & -- \\
        ChvatalGraph & 12 & 24 & 6 & 49.981 & 0.115 & 0.113 \\
        CirculantGraph\_20\_5 & 20 & 20 & 2 & 0.42 & 0.105 & 0.107 \\
        CircularLadderGraph\_20 & 40 & 60 & 4 & 69.035 & 0.327 & 0.223 \\
        ClebschGraph & 16 & 40 & 8 & -- & 0.126 & 0.123 \\
        CompleteBipartiteGraph\_25\_20 & 45 & 500 & 20 & -- & 0.244 & 0.225 \\
        CompleteGraph\_15 & 15 & 105 & 10 & -- & 0.101 & 0.099 \\
        CoxeterGraph & 28 & 42 & 7 & -- & 1.215 & 0.297 \\
        CubeGraph\_6 & 64 & 192 & -- & -- & -- & -- \\
        CycleGraph\_100 & 100 & 100 & 2 & 98.221 & 1.411 & 0.592 \\
        DejterGraph & 112 & 336 & -- & -- & -- & -- \\
        DesarguesGraph & 20 & 30 & 6 & 123.133 & 0.204 & 0.159 \\
        DodecahedralGraph & 20 & 30 & 6 & 239.696 & 0.261 & 0.215 \\
        DorogovtsevGoltsevMendesGraph & 3282 & 6561 & -- & -- & -- & -- \\
        DoubleStarSnark & 30 & 45 & 6 & -- & 0.501 & 0.225 \\
        DurerGraph & 12 & 18 & 4 & 0.655 & 0.115 & 0.116 \\
        DyckGraph & 32 & 48 & 8 & -- & 5.05 & 1.525 \\
        EllinghamHorton54Graph & 54 & 81 & 6 & -- & 1.575 & 0.531 \\
        EllinghamHorton78Graph & 78 & 117 & 6 & -- & 2.364 & 1.079 \\
        ErreraGraph & 17 & 45 & 6 & -- & 0.108 & 0.118 \\
        F26AGraph & 26 & 39 & 7 & -- & 1.135 & 0.407 \\
        FibonacciTree\_10 & 143 & 142 & 2 & -- & 0.105 & 0.108 \\
        FlowerSnark & 20 & 30 & 6 & 158.722 & 0.224 & 0.155 \\
        FoldedCubeGraph\_7 & 64 & 224 & -- & -- & -- & -- \\
        FolkmanGraph & 20 & 40 & 6 & 430.357 & 0.127 & 0.128 \\
        FosterGraph & 90 & 135 & -- & -- & -- & -- \\
        FranklinGraph & 12 & 18 & 4 & 0.598 & 0.114 & 0.113 \\
        FriendshipGraph\_10 & 21 & 30 & 2 & 0.876 & 0.093 & 0.107 \\
        FruchtGraph & 12 & 18 & 3 & 0.452 & 0.108 & 0.107 \\
        GeneralizedPetersenGraph\_10\_4 & 20 & 30 & 6 & 133.92 & 0.212 & 0.164 \\
        GoethalsSeidelGraph\_2\_3 & 16 & 72 & 10 & -- & 0.118 & 0.121 \\
        \midrule
        \multicolumn{7}{r}{\small \slshape Continued on next page} \\
        \bottomrule
    \end{tabular}
\end{table}

\begin{table}[p]
    \centering
    \caption[]{Performance comparison on named graphs (2/4) (continued).}
    \small
    \begin{tabular}{lrrrrrr}
        \toprule
        Instance & $|V|$ & $|E|$ & $\textrm{bw}$ & SAT & Alg 1 & Alg 2 \\ \midrule
        GoldnerHararyGraph & 11 & 27 & 4 & 2.444 & 0.102 & 0.106 \\
        GossetGraph & 56 & 756 & 36 & -- & 0.638 & 0.833 \\
        GrayGraph & 54 & 81 & 9 & -- & 489.292 & 19.648 \\
        Grid2dGraph\_20\_40 & 800 & 1540 & -- & -- & -- & -- \\
        Grid2dGraph\_5\_5 & 25 & 40 & 5 & 233.19 & 0.188 & 0.159 \\
        GrotzschGraph & 11 & 20 & 5 & 3.345 & 0.11 & 0.112 \\
        HallJankoGraph & 100 & 1800 & -- & -- & -- & -- \\
        HanoiTowerGraph\_4\_3 & 64 & 168 & 12 & -- & 494.091 & 49.575 \\
        HararyGraph\_6\_15 & 15 & 45 & 6 & -- & 0.114 & 0.116 \\
        HarborthGraph & 52 & 104 & 4 & -- & 0.147 & 0.139 \\
        HarriesGraph & 70 & 105 & -- & -- & -- & -- \\
        HarriesWongGraph & 70 & 105 & -- & -- & -- & -- \\
        HeawoodGraph & 14 & 21 & 5 & 3.103 & 0.125 & 0.123 \\
        HerschelGraph & 11 & 18 & 4 & 0.642 & 0.11 & 0.107 \\
        HexahedralGraph & 8 & 12 & 4 & 0.407 & 0.106 & 0.11 \\
        HigmanSimsGraph & 100 & 1100 & -- & -- & -- & -- \\
        HoffmanGraph & 16 & 32 & 6 & 111.353 & 0.123 & 0.122 \\
        HoffmanSingletonGraph & 50 & 175 & -- & -- & -- & -- \\
        HoltGraph & 27 & 54 & 9 & -- & 0.844 & 0.259 \\
        HortonGrap & 96 & 144 & 6 & -- & 3.26 & 1.602 \\
        HouseGraph & 5 & 6 & 2 & 0.327 & 0.102 & 0.105 \\
        HouseXGraph & 5 & 8 & 3 & 0.342 & 0.098 & 0.101 \\
        HyperStarGraph\_10\_2 & 45 & 72 & 6 & -- & -- & 0.158 \\
        HyperStarGraph\_10\_5 & 252 & 630 & -- & -- & -- & -- \\
        IcosahedralGraph & 12 & 30 & 6 & 158.176 & 0.115 & 0.116 \\
        JohnsonGraph\_10\_4 & 210 & 2520 & -- & -- & -- & -- \\
        JohnsonGraph\_8\_2 & 28 & 168 & 15 & -- & 0.141 & 0.144 \\
        KittellGraph & 23 & 63 & 6 & -- & 0.123 & 0.123 \\
        Klein3RegularGraph & 56 & 84 & -- & -- & -- & -- \\
        Klein7RegularGraph & 24 & 84 & 12 & -- & 0.241 & 0.223 \\
        KneserGraph\_10\_2 & 45 & 630 & 30 & -- & 0.431 & 0.585 \\
        KneserGraph\_8\_3 & 56 & 280 & -- & -- & -- & -- \\
        KrackhardtKiteGraph & 10 & 18 & 3 & 0.464 & 0.107 & 0.106 \\
        LadderGraph\_20 & 40 & 58 & 2 & 38.487 & 0.116 & 0.115 \\
        LjubljanaGraph & 112 & 168 & -- & -- & -- & -- \\
        LollipopGraph\_7\_5 & 12 & 26 & 5 & 18.837 & 0.1 & 0.089 \\
        M22Graph & 77 & 616 & -- & -- & -- & -- \\
        MarkstroemGraph & 24 & 36 & 4 & 13.61 & 0.15 & 0.137 \\
        McGeeGraph & 24 & 36 & 7 & -- & 0.929 & 0.282 \\
        MeredithGraph & 70 & 140 & 6 & -- & 1.195 & 0.57 \\
        MoebiusKantorGraph & 16 & 24 & 6 & 58.467 & 0.191 & 0.171 \\
        MoserSpindle & 7 & 11 & 3 & 0.354 & 0.104 & 0.091 \\
        NauruGraph & 24 & 36 & 6 & 338.777 & 0.316 & 0.163 \\
        NKStarGraph\_5\_3 & 60 & 120 & 12 & -- & -- & 88.064 \\
        NKStarGraph\_8\_2 & 56 & 196 & -- & -- & -- & -- \\
        \midrule
        \multicolumn{7}{r}{\small \slshape Continued on next page} \\
        \bottomrule
    \end{tabular}
\end{table}

\begin{table}[p]
    \centering
    \caption[]{Performance comparison on named graphs (3/4) (continued).}
    \small
    \begin{tabular}{lrrrrrr}
        \toprule
        Instance & $|V|$ & $|E|$ & $\textrm{bw}$ & SAT & Alg 1 & Alg 2 \\ \midrule
        NonisotropicOrthogonalPolarGraph\_3\_5 & 15 & 60 & 8 & -- & 0.109 & 0.109 \\
        NonisotropicUnitaryPolarGraph\_3\_3 & 63 & 1008 & 40 & -- & 0.611 & 0.978 \\
        NStarGraph\_5 & 120 & 240 & -- & -- & -- & -- \\
        OctahedralGraph & 6 & 12 & 4 & 0.411 & 0.103 & 0.104 \\
        OddGraph\_3 & 10 & 15 & 4 & 0.508 & 0.101 & 0.1 \\
        OddGraph\_4 & 35 & 70 & 12 & -- & 88.59 & 5.502 \\
        OddGraph\_5 & 126 & 315 & -- & -- & -- & -- \\
        OrthogonalArrayBlockGraph\_4\_3 & 9 & 36 & 6 & -- & 0.092 & 0.089 \\
        OrthogonalPolarGraph\_5\_2 & 15 & 45 & 9 & -- & 0.126 & 0.138 \\
        PaleyGraph\_17 & 17 & 68 & 10 & -- & 0.122 & 0.121 \\
        PappusGraph & 18 & 27 & 6 & 86.388 & 0.21 & 0.147 \\
        PasechnikGraph\_1 & 9 & 9 & 2 & 0.331 & 0.1 & 0.089 \\
        PasechnikGraph\_2 & 49 & 441 & 30 & -- & 3.739 & 6.75 \\
        PasechnikGraph\_3 & 121 & 3025 & -- & -- & -- & -- \\
        PathGraph\_100 & 100 & 99 & 2 & 94.348 & 0.104 & 0.105 \\
        PerkelGraph & 57 & 171 & -- & -- & -- & -- \\
        PetersenGraph & 10 & 15 & 4 & 0.513 & 0.111 & 0.113 \\
        PoussinGraph & 15 & 39 & 6 & -- & 0.117 & 0.115 \\
        RingedTree\_10 & 1023 & 2043 & -- & -- & -- & -- \\
        RingedTree\_6 & 63 & 123 & 7 & -- & 4.485 & 1.131 \\
        RingedTree\_8 & 255 & 507 & -- & -- & -- & -- \\
        RobertsonGraph & 19 & 38 & 8 & -- & 0.168 & 0.144 \\
        SchlaefliGraph & 27 & 216 & 16 & -- & 0.123 & 0.128 \\
        ShrikhandeGraph & 16 & 48 & 8 & -- & 0.12 & 0.121 \\
        SierpinskiGasketGraph\_3 & 15 & 27 & 3 & 0.818 & 0.106 & 0.109 \\
        SierpinskiGasketGraph\_5 & 123 & 243 & 4 & -- & 0.473 & 0.291 \\
        SimsGewirtzGraph & 56 & 280 & -- & -- & -- & -- \\
        SousselierGraph & 16 & 27 & 5 & 5.528 & 0.124 & 0.118 \\
        SquaredSkewHadamardMatrixGraph\_1 & 9 & 18 & 4 & 0.616 & 0.104 & 0.104 \\
        SquaredSkewHadamardMatrixGraph\_2 & 49 & 588 & 32 & -- & 0.501 & 0.895 \\
        SquaredSkewHadamardMatrixGraph\_3 & 121 & 3630 & 80 & -- & 10.039 & 51.771 \\
        StarGraph\_100 & 101 & 100 & 1 & -- & 0.094 & 0.095 \\
        SwitchedSquaredSkewHadamardMatrixGraph\_1 & 10 & 15 & 4 & 0.479 & 0.11 & 0.111 \\
        SwitchedSquaredSkewHadamardMatrixGraph\_2 & 50 & 525 & 32 & -- & 0.99 & 1.203 \\
        SwitchedSquaredSkewHadamardMatrixGraph\_3 & 122 & 3355 & 80 & -- & 18.118 & 54.602 \\
        SylvesterGraph & 36 & 90 & 15 & -- & 378.549 & 117.714 \\
        SymplecticDualPolarGraph\_4\_3 & 40 & 240 & 23 & -- & 6.651 & 8.028 \\
        SymplecticDualPolarGraph\_4\_4 & 85 & 850 & -- & -- & -- & -- \\
        SymplecticPolarGraph\_4\_3 & 40 & 240 & 24 & -- & 8.548 & 24.009 \\
        SymplecticPolarGraph\_4\_4 & 85 & 850 & -- & -- & -- & -- \\
        SzekeresSnarkGraph & 50 & 75 & 6 & -- & 1.685 & 0.544 \\
        T2starGeneralizedQuadrangleGraph\_2 & 8 & 16 & 4 & 0.478 & 0.103 & 0.107 \\
        T2starGeneralizedQuadrangleGraph\_4 & 64 & 576 & -- & -- & -- & -- \\
        TaylorTwographDescendantSRG\_3 & 27 & 135 & 16 & -- & 0.248 & 0.384 \\
        TaylorTwographSRG\_3 & 28 & 210 & 18 & -- & 0.201 & 0.203 \\
        \midrule
        \multicolumn{7}{r}{\small \slshape Continued on next page} \\
        \bottomrule
    \end{tabular}
\end{table}

\begin{table}[p]
    \centering
    \caption[]{Performance comparison on named graphs (4/4) (continued).}
    \small
    \begin{tabular}{lrrrrrr}
        \toprule
        Instance & $|V|$ & $|E|$ & $\textrm{bw}$ & SAT & Alg 1 & Alg 2 \\ \midrule
        TetrahedralGraph & 4 & 6 & 3 & 0.333 & 0.1 & 0.099 \\
        ThomsenGraph & 6 & 9 & 3 & 0.344 & 0.103 & 0.103 \\
        TietzeGraph & 12 & 18 & 4 & 0.662 & 0.113 & 0.115 \\
        Toroidal6RegularGrid2dGraph\_4\_6 & 24 & 72 & 8 & -- & 0.138 & 0.135 \\
        TruncatedIcosidodecahedralGraph & 120 & 180 & -- & -- & -- & -- \\
        TruncatedTetrahedralGraph & 12 & 18 & 4 & 0.913 & 0.118 & 0.117 \\
        Tutte12Cage & 126 & 189 & -- & -- & -- & -- \\
        TutteCoxeterGraph & 30 & 45 & 8 & -- & 8.308 & 1.037 \\
        TutteGraph & 46 & 69 & 5 & -- & 0.35 & 0.267 \\
        WagnerGraph & 8 & 12 & 4 & 0.401 & 0.104 & 0.11 \\
        WatkinsSnarkGraph & 50 & 75 & 6 & -- & 1.731 & 0.624 \\
        WellsGraph & 32 & 80 & 14 & -- & 34.231 & 26.103 \\
        WheelGraph\_100 & 100 & 198 & 3 & -- & 1.399 & 0.592 \\
        WienerArayaGraph & 42 & 67 & 7 & -- & 1.429 & 0.658 \\
        WorldMap & 166 & 323 & 5 & -- & 0.798 & 0.324 \\
        \bottomrule
    \end{tabular}
\end{table}

\begin{table}[p]
    \centering
    \caption{Performance comparison on treedepth instances (1/5).}
    \label{tab:results_td}
    \small
    \begin{tabular}{lrrrrrr}
        \toprule
        Instance & $|V|$ & $|E|$ & $\textrm{bw}$ & SAT & Alg 1 & Alg 2 \\ \midrule
        exact\_001.gr & 10 & 15 & 4 & 0.511 & 0.111 & 0.111 \\
        exact\_002.gr & 16 & 58 & 8 & -- & 0.113 & 0.116 \\
        exact\_003.gr & 17 & 62 & 7 & -- & 0.105 & 0.108 \\
        exact\_004.gr & 18 & 18 & 2 & 0.371 & 0.102 & 0.102 \\
        exact\_005.gr & 20 & 21 & 2 & 0.408 & 0.105 & 0.104 \\
        exact\_006.gr & 20 & 24 & 2 & 0.475 & 0.095 & 0.107 \\
        exact\_007.gr & 20 & 30 & 6 & 157.751 & 0.243 & 0.157 \\
        exact\_008.gr & 20 & 39 & 3 & 3.374 & 0.107 & 0.108 \\
        exact\_009.gr & 21 & 25 & 3 & 0.942 & 0.112 & 0.113 \\
        exact\_010.gr & 21 & 72 & 8 & -- & 0.111 & 0.112 \\
        exact\_011.gr & 23 & 24 & 2 & 0.512 & 0.105 & 0.106 \\
        exact\_012.gr & 24 & 28 & 2 & 0.673 & 0.107 & 0.108 \\
        exact\_013.gr & 25 & 46 & 3 & 7.62 & 0.109 & 0.108 \\
        exact\_014.gr & 25 & 52 & 5 & 203.209 & 0.117 & 0.116 \\
        exact\_015.gr & 26 & 30 & 2 & 0.769 & 0.103 & 0.104 \\
        exact\_016.gr & 26 & 51 & 3 & 10.494 & 0.105 & 0.111 \\
        exact\_017.gr & 26 & 64 & 7 & -- & 0.131 & 0.121 \\
        exact\_018.gr & 27 & 87 & 7 & -- & 0.12 & 0.118 \\
        exact\_019.gr & 28 & 57 & 4 & 48.094 & 0.116 & 0.115 \\
        exact\_020.gr & 28 & 168 & 15 & -- & 0.138 & 0.142 \\
        exact\_021.gr & 29 & 38 & 2 & 1.89 & 0.117 & 0.11 \\
        exact\_022.gr & 29 & 47 & 3 & 9.454 & 0.112 & 0.113 \\
        exact\_023.gr & 30 & 38 & 3 & 2.805 & 0.12 & 0.118 \\
        exact\_024.gr & 30 & 45 & 5 & 100.297 & 0.293 & 0.186 \\
        exact\_025.gr & 30 & 51 & 4 & 48.875 & 0.122 & 0.119 \\
        exact\_026.gr & 30 & 59 & 4 & 561.575 & 0.116 & 0.119 \\
        exact\_027.gr & 30 & 70 & 6 & -- & 0.134 & 0.126 \\
        exact\_028.gr & 31 & 64 & 4 & 385.584 & 0.123 & 0.12 \\
        exact\_029.gr & 32 & 119 & 8 & -- & 0.118 & 0.118 \\
        exact\_030.gr & 32 & 137 & 12 & -- & 0.24 & 0.171 \\
        exact\_031.gr & 33 & 78 & 7 & -- & 0.135 & 0.134 \\
        exact\_032.gr & 34 & 78 & 5 & -- & 0.133 & 0.124 \\
        exact\_033.gr & 34 & 82 & 8 & -- & 0.137 & 0.13 \\
        exact\_034.gr & 35 & 42 & 3 & 5.872 & 0.134 & 0.126 \\
        exact\_035.gr & 35 & 70 & 3 & 105.028 & 0.114 & 0.116 \\
        exact\_036.gr & 37 & 68 & 3 & 23.392 & 0.112 & 0.104 \\
        exact\_037.gr & 39 & 48 & 3 & 8.326 & 0.113 & 0.11 \\
        exact\_038.gr & 39 & 62 & 3 & 26.409 & 0.119 & 0.117 \\
        exact\_039.gr & 40 & 54 & 3 & 13.121 & 0.129 & 0.127 \\
        exact\_040.gr & 40 & 73 & 4 & -- & 0.139 & 0.127 \\
        exact\_041.gr & 40 & 79 & 4 & -- & 0.133 & 0.123 \\
        exact\_042.gr & 40 & 112 & 7 & -- & 0.128 & 0.127 \\
        exact\_043.gr & 40 & 129 & 8 & -- & 0.177 & 0.139 \\
        exact\_044.gr & 41 & 52 & 3 & 24.177 & 0.131 & 0.126 \\
        exact\_045.gr & 41 & 58 & 6 & -- & 1.539 & 0.469 \\
        \midrule
        \multicolumn{7}{r}{\small \slshape Continued on next page} \\
        \bottomrule
    \end{tabular}
\end{table}

\begin{table}[p]
    \centering
    \caption[]{Performance comparison on treedepth instances (2/5) (continued).}
    \small
    \begin{tabular}{lrrrrrr}
        \toprule
        Instance & $|V|$ & $|E|$ & $\textrm{bw}$ & SAT & Alg 1 & Alg 2 \\ \midrule
        exact\_046.gr & 41 & 76 & 4 & -- & 0.124 & 0.117 \\
        exact\_047.gr & 41 & 256 & 19 & -- & 1.055 & 0.388 \\
        exact\_048.gr & 42 & 67 & 7 & -- & 1.436 & 0.667 \\
        exact\_049.gr & 42 & 73 & 6 & -- & 0.912 & 0.243 \\
        exact\_050.gr & 43 & 53 & 3 & 13.815 & 0.152 & 0.139 \\
        exact\_051.gr & 45 & 91 & 5 & -- & 0.175 & 0.145 \\
        exact\_052.gr & 45 & 630 & 30 & 493.793 & 0.436 & 0.563 \\
        exact\_053.gr & 46 & 63 & 3 & 53.679 & 0.115 & 0.122 \\
        exact\_054.gr & 47 & 67 & 4 & 160.085 & 0.178 & 0.142 \\
        exact\_055.gr & 49 & 85 & 4 & -- & 0.125 & 0.119 \\
        exact\_056.gr & 50 & 59 & 3 & 83.7 & 0.117 & 0.118 \\
        exact\_057.gr & 50 & 75 & 7 & -- & 8.855 & 1.251 \\
        exact\_058.gr & 50 & 98 & 5 & -- & 0.176 & 0.142 \\
        exact\_059.gr & 50 & 525 & 32 & 366.016 & 1.015 & 1.127 \\
        exact\_060.gr & 51 & 151 & 8 & -- & 0.173 & 0.143 \\
        exact\_061.gr & 53 & 100 & 7 & -- & 0.552 & 0.239 \\
        exact\_062.gr & 54 & 66 & 4 & -- & 0.228 & 0.17 \\
        exact\_063.gr & 55 & 100 & 5 & -- & 0.256 & 0.183 \\
        exact\_064.gr & 57 & 79 & 3 & 87.537 & 0.138 & 0.131 \\
        exact\_065.gr & 57 & 153 & 6 & -- & 0.165 & 0.139 \\
        exact\_066.gr & 59 & 154 & 6 & -- & 0.139 & 0.134 \\
        exact\_067.gr & 60 & 72 & 3 & 107.222 & 0.14 & 0.132 \\
        exact\_068.gr & 60 & 90 & 10 & -- & -- & 552.296 \\
        exact\_069.gr & 60 & 119 & 4 & -- & 0.168 & 0.139 \\
        exact\_070.gr & 60 & 176 & 10 & -- & 7.414 & 1.294 \\
        exact\_071.gr & 61 & 151 & 10 & -- & 0.264 & 0.169 \\
        exact\_072.gr & 62 & 108 & 4 & -- & 0.159 & 0.138 \\
        exact\_073.gr & 62 & 159 & 9 & -- & 0.693 & 0.266 \\
        exact\_074.gr & 63 & 1008 & 40 & -- & 0.649 & 0.998 \\
        exact\_075.gr & 64 & 168 & 12 & -- & 498.103 & 48.922 \\
        exact\_076.gr & 65 & 128 & 9 & -- & 0.82 & 0.24 \\
        exact\_077.gr & 66 & 120 & 5 & -- & 0.301 & 0.184 \\
        exact\_078.gr & 67 & 152 & 6 & -- & 0.19 & 0.147 \\
        exact\_079.gr & 68 & 83 & 3 & 191.805 & 0.14 & 0.134 \\
        exact\_080.gr & 68 & 116 & 3 & -- & 0.122 & 0.12 \\
        exact\_081.gr & 70 & 79 & 2 & 39.004 & 0.123 & 0.12 \\
        exact\_082.gr & 70 & 139 & 4 & -- & 0.3 & 0.134 \\
        exact\_083.gr & 70 & 274 & 12 & 249.498 & 1.528 & 0.557 \\
        exact\_084.gr & 71 & 103 & -- & -- & -- & -- \\
        exact\_085.gr & 72 & 158 & 5 & -- & 0.187 & 0.168 \\
        exact\_086.gr & 75 & 118 & 4 & -- & 0.137 & 0.199 \\
        exact\_087.gr & 75 & 304 & 8 & 270.882 & 0.184 & 0.151 \\
        exact\_088.gr & 77 & 146 & 10 & -- & -- & 135.583 \\
        exact\_089.gr & 77 & 254 & 8 & 254.525 & 0.291 & 0.163 \\
        exact\_090.gr & 77 & 616 & -- & 474.88 & -- & -- \\
        \midrule
        \multicolumn{7}{r}{\small \slshape Continued on next page} \\
        \bottomrule
    \end{tabular}
\end{table}

\begin{table}[p]
    \centering
    \caption[]{Performance comparison on treedepth instances (3/5) (continued).}
    \small
    \begin{tabular}{lrrrrrr}
        \toprule
        Instance & $|V|$ & $|E|$ & $\textrm{bw}$ & SAT & Alg 1 & Alg 2 \\ \midrule
        exact\_091.gr & 80 & 98 & 3 & -- & 0.143 & 0.139 \\
        exact\_092.gr & 80 & 120 & -- & -- & -- & -- \\
        exact\_093.gr & 80 & 159 & 4 & -- & 0.321 & 0.16 \\
        exact\_094.gr & 81 & 810 & -- & -- & -- & -- \\
        exact\_095.gr & 82 & 100 & 2 & 59.071 & 0.169 & 0.15 \\
        exact\_096.gr & 84 & 339 & 9 & -- & 0.414 & 0.238 \\
        exact\_097.gr & 86 & 123 & 4 & -- & 0.156 & 0.132 \\
        exact\_098.gr & 87 & 406 & 12 & -- & 0.994 & 0.324 \\
        exact\_099.gr & 88 & 133 & 5 & -- & 1.43 & 0.637 \\
        exact\_100.gr & 90 & 135 & -- & -- & -- & -- \\
        exact\_101.gr & 90 & 135 & -- & -- & -- & -- \\
        exact\_102.gr & 92 & 104 & 3 & -- & 0.174 & 0.153 \\
        exact\_103.gr & 92 & 131 & 6 & -- & 8.079 & 1.63 \\
        exact\_104.gr & 92 & 285 & 8 & -- & 0.333 & 0.2 \\
        exact\_105.gr & 93 & 111 & 4 & -- & 0.231 & 0.174 \\
        exact\_106.gr & 94 & 118 & 4 & -- & 0.269 & 0.191 \\
        exact\_107.gr & 95 & 121 & 5 & -- & 2.265 & 0.702 \\
        exact\_108.gr & 96 & 153 & 8 & -- & 29.653 & 2.06 \\
        exact\_109.gr & 98 & 109 & 3 & -- & 0.222 & 0.169 \\
        exact\_110.gr & 99 & 119 & 3 & -- & 0.194 & 0.162 \\
        exact\_111.gr & 100 & 384 & 9 & -- & 0.581 & 0.246 \\
        exact\_112.gr & 100 & 1800 & -- & -- & -- & -- \\
        exact\_113.gr & 102 & 113 & 4 & -- & 0.125 & 0.133 \\
        exact\_114.gr & 102 & 142 & 4 & -- & 0.177 & 0.138 \\
        exact\_115.gr & 102 & 144 & 4 & -- & 0.384 & 0.265 \\
        exact\_116.gr & 105 & 312 & 8 & -- & 0.367 & 0.217 \\
        exact\_117.gr & 105 & 441 & 12 & -- & 6.447 & 1.378 \\
        exact\_118.gr & 110 & 219 & 4 & -- & 0.347 & 0.153 \\
        exact\_119.gr & 110 & 223 & 11 & -- & -- & 57.89 \\
        exact\_120.gr & 111 & 1029 & 27 & -- & 148.875 & 18.422 \\
        exact\_121.gr & 112 & 168 & -- & -- & -- & -- \\
        exact\_122.gr & 112 & 425 & -- & -- & -- & -- \\
        exact\_123.gr & 113 & 164 & 5 & -- & 3.239 & 1.241 \\
        exact\_124.gr & 113 & 168 & -- & -- & -- & -- \\
        exact\_125.gr & 115 & 161 & 4 & -- & 0.192 & 0.144 \\
        exact\_126.gr & 117 & 146 & 4 & -- & 0.182 & 0.154 \\
        exact\_127.gr & 119 & 147 & 4 & -- & 0.138 & 0.126 \\
        exact\_128.gr & 120 & 157 & 5 & -- & 5.073 & 1.615 \\
        exact\_129.gr & 120 & 240 & -- & -- & -- & -- \\
        exact\_130.gr & 120 & 479 & 6 & -- & 2.442 & 0.621 \\
        exact\_131.gr & 121 & 220 & -- & -- & -- & -- \\
        exact\_132.gr & 123 & 243 & 4 & -- & 0.476 & 0.315 \\
        exact\_133.gr & 123 & 409 & 8 & -- & 0.996 & 0.284 \\
        exact\_134.gr & 125 & 141 & 4 & -- & 17.289 & 3.517 \\
        exact\_135.gr & 126 & 189 & -- & -- & -- & -- \\
        \midrule
        \multicolumn{7}{r}{\small \slshape Continued on next page} \\
        \bottomrule
    \end{tabular}
\end{table}

\begin{table}[p]
    \centering
    \caption[]{Performance comparison on treedepth instances (4/5) (continued).}
    \small
    \begin{tabular}{lrrrrrr}
        \toprule
        Instance & $|V|$ & $|E|$ & $\textrm{bw}$ & SAT & Alg 1 & Alg 2 \\ \midrule
        exact\_136.gr & 126 & 315 & -- & -- & -- & -- \\
        exact\_137.gr & 126 & 4100 & 47 & -- & -- & 43.924 \\
        exact\_138.gr & 127 & 189 & -- & -- & -- & -- \\
        exact\_139.gr & 128 & 266 & 11 & -- & -- & 275.225 \\
        exact\_140.gr & 131 & 193 & -- & -- & -- & -- \\
        exact\_141.gr & 132 & 191 & 5 & -- & 4.299 & 1.29 \\
        exact\_142.gr & 133 & 167 & 4 & -- & 0.248 & 0.202 \\
        exact\_143.gr & 138 & 447 & 8 & -- & 0.79 & 0.306 \\
        exact\_144.gr & 138 & 493 & 12 & -- & -- & 1.323 \\
        exact\_145.gr & 140 & 279 & 5 & -- & 0.957 & 0.309 \\
        exact\_146.gr & 141 & 214 & -- & -- & -- & -- \\
        exact\_147.gr & 142 & 196 & 5 & -- & 4.808 & 1.461 \\
        exact\_148.gr & 145 & 2512 & 44 & -- & 297.415 & 11.615 \\
        exact\_149.gr & 146 & 250 & 7 & -- & 45.87 & 6.271 \\
        exact\_150.gr & 148 & 198 & 5 & -- & 6.98 & 1.939 \\
        exact\_151.gr & 150 & 599 & 6 & -- & 2.026 & 0.258 \\
        exact\_152.gr & 151 & 223 & -- & -- & -- & -- \\
        exact\_153.gr & 152 & 476 & 9 & -- & 5.739 & 1.251 \\
        exact\_154.gr & 157 & 3541 & 31 & -- & -- & 278.625 \\
        exact\_155.gr & 160 & 240 & -- & -- & -- & -- \\
        exact\_156.gr & 160 & 319 & 4 & -- & 0.533 & 0.274 \\
        exact\_157.gr & 163 & 195 & 4 & -- & 0.744 & 0.319 \\
        exact\_158.gr & 163 & 278 & 8 & -- & -- & 63.475 \\
        exact\_159.gr & 165 & 522 & 9 & -- & 4.06 & 1.089 \\
        exact\_160.gr & 167 & 213 & 5 & -- & 0.375 & 0.211 \\
        exact\_161.gr & 167 & 328 & 5 & -- & 0.79 & 0.327 \\
        exact\_162.gr & 170 & 218 & 6 & -- & 316.127 & 25.113 \\
        exact\_163.gr & 170 & 255 & -- & -- & -- & -- \\
        exact\_164.gr & 170 & 679 & 6 & -- & 2.949 & 1.369 \\
        exact\_165.gr & 176 & 186 & 3 & -- & 0.531 & 0.303 \\
        exact\_166.gr & 176 & 3973 & 29 & -- & 3.252 & 1.333 \\
        exact\_167.gr & 180 & 270 & -- & -- & -- & -- \\
        exact\_168.gr & 180 & 359 & 5 & -- & 2.145 & 0.368 \\
        exact\_169.gr & 181 & 253 & 6 & -- & 60.95 & 7.893 \\
        exact\_170.gr & 184 & 268 & -- & -- & -- & -- \\
        exact\_171.gr & 191 & 283 & -- & -- & -- & -- \\
        exact\_172.gr & 195 & 342 & 6 & -- & 27.777 & 4.688 \\
        exact\_173.gr & 198 & 692 & 9 & -- & 8.168 & 1.318 \\
        exact\_174.gr & 199 & 265 & 5 & -- & 18.844 & 3.585 \\
        exact\_175.gr & 200 & 300 & -- & -- & -- & -- \\
        exact\_176.gr & 200 & 799 & 6 & -- & 22.262 & 6.238 \\
        exact\_177.gr & 204 & 248 & 4 & -- & 0.942 & 0.26 \\
        exact\_178.gr & 210 & 2520 & -- & -- & -- & -- \\
        exact\_179.gr & 212 & 257 & 5 & -- & 6.501 & 1.964 \\
        exact\_180.gr & 214 & 785 & 9 & -- & 9.719 & 1.574 \\
        \midrule
        \multicolumn{7}{r}{\small \slshape Continued on next page} \\
        \bottomrule
    \end{tabular}
\end{table}

\begin{table}[p]
    \centering
    \caption[]{Performance comparison on treedepth instances (5/5) (continued).}\label{tab:last}
    \small
    \begin{tabular}{lrrrrrr}
        \toprule
        Instance & $|V|$ & $|E|$ & $\textrm{bw}$ & SAT & Alg 1 & Alg 2 \\ \midrule
        exact\_181.gr & 223 & 265 & 4 & -- & 0.495 & 0.211 \\
        exact\_182.gr & 225 & 771 & 9 & -- & 1.713 & 0.769 \\
        exact\_183.gr & 231 & 3465 & -- & -- & -- & -- \\
        exact\_184.gr & 244 & 290 & 4 & -- & 0.544 & 0.217 \\
        exact\_185.gr & 276 & 1187 & 10 & -- & 10.02 & 1.928 \\
        exact\_186.gr & 280 & 360 & 4 & -- & 5.347 & 1.277 \\
        exact\_187.gr & 300 & 450 & -- & -- & -- & -- \\
        exact\_188.gr & 300 & 508 & -- & -- & -- & -- \\
        exact\_189.gr & 300 & 1199 & 6 & -- & 35.1 & 7.161 \\
        exact\_190.gr & 352 & 440 & 6 & -- & -- & 62.734 \\
        exact\_191.gr & 439 & 873 & -- & -- & -- & -- \\
        exact\_192.gr & 441 & 1638 & -- & -- & -- & -- \\
        exact\_193.gr & 449 & 2213 & 13 & -- & 488.196 & 40.182 \\
        exact\_194.gr & 450 & 1799 & 7 & -- & 595.01 & 232.541 \\
        exact\_195.gr & 451 & 587 & -- & -- & -- & -- \\
        exact\_196.gr & 457 & 684 & -- & -- & -- & -- \\
        exact\_197.gr & 469 & 614 & 8 & -- & -- & 145.815 \\
        exact\_198.gr & 479 & 638 & -- & -- & -- & -- \\
        exact\_199.gr & 491 & 645 & -- & -- & -- & -- \\
        exact\_200.gr & 498 & 702 & -- & -- & -- & -- \\
        \bottomrule
    \end{tabular}
\end{table}

\end{document}